\documentclass[11pt]{article}
\usepackage{hyperref}
\usepackage{verbatim}
\usepackage{fullpage}
\usepackage{amsmath}
\usepackage{latexsym}
\usepackage{revsymb}
\usepackage{yfonts}

\usepackage{natbib}
\usepackage{amsfonts}
\usepackage{amsmath}
\usepackage{amssymb}
\usepackage{amsthm}
\usepackage{graphicx}
\usepackage{bm}
\usepackage{bbm}

\newcommand{\be}{\begin{eqnarray} \begin{aligned}}
\newcommand{\ee}{\end{aligned} \end{eqnarray} }
\newcommand{\benn}{\begin{eqnarray*} \begin{aligned}}
\newcommand{\eenn}{\end{aligned} \end{eqnarray*} }

\newcommand{\ec}{\end{center}}
\newcommand{\half}{\frac{1}{2}}

\newcommand{\id}{\mathbb{I}}

\newcommand{\Tr}{\mathop{\mathrm{Tr}}\nolimits}

\newcommand{\bigoh}{\mathrm{O}}

\newtheorem{theorem}{Theorem}[section]
\newtheorem{examp}[theorem]{Example}
\newtheorem{lemma}[theorem]{Lemma}
\newtheorem{claim}[theorem]{Claim}
\newtheorem{definition}[theorem]{Definition}

\newtheorem{corollary}[theorem]{Corollary}

\newcommand{\mc}{\mathcal}

\usepackage{amsfonts}
\def\Real{\mathbb{R}}
\def\Complex{\mathbb{C}}

\def\Integer{\mathbb{Z}}
\def\id{\mathbb{I}}

\def\01{\{0,1\}}

\newcommand{\eps}{\varepsilon}
\newcommand{\ket}[1]{|#1\rangle}
\newcommand{\bra}[1]{\langle#1|}
\newcommand{\outp}[2]{|#1\rangle\langle#2|}

\newcommand{\CM}{C}

\newcommand{\setS}{\mathcal{S}}
\newcommand{\setA}{\mathcal{A}}
\newcommand{\setAC}{\setA^{\times |C|}}
\newcommand{\setM}{\mathcal{O}}
\newcommand{\setE}{\mathcal{E}}

\newcommand{\spdBin}{\mathcal{S}_p^d}

\newcommand{\polylog}{\mbox{polylog}}

\newcommand{\assign}{:=}

\newcommand{\fat}{\mathrm{fat}}

\newenvironment{sdp}[2]{
\smallskip
\begin{center}
\begin{tabular}{ll}
#1 & #2\\
subject to
}
{
\end{tabular}
\end{center}
\smallskip
}

\newcommand{\cancel}[1]{}

\begin{document}

\title{{\sf Relaxed uncertainty relations and information processing}\\
\begin{center}
\large{Greg Ver Steeg\footnote{gregv@caltech.edu} \hspace{1cm} Stephanie Wehner\footnote{wehner@caltech.edu}\\
\emph{Institute for Quantum Information, California Institute of Technology, \\Pasadena CA 91125, USA}}
\end{center}\vspace{-0.4cm}
}
\date{\today}
\maketitle
\begin{abstract}
We consider a range of ``theories'' that violate the uncertainty relation for anti-commuting
observables derived in [JMP, 49, 062105 (2008)]. We first show that Tsirelson's bound
for the CHSH inequality can be derived from this uncertainty relation, and that relaxing
this relation allows for non-local correlations that are stronger than what can be obtained in
quantum mechanics. We continue to construct a hierarchy of related non-signaling theories, and show
that on one hand they admit superstrong random access encodings and exponential savings for a particular
communication problem, while on the other hand it becomes much harder in these theories to learn a 
state. We show that the existence of these effects stems from the absence of certain
constraints on the expectation values of \emph{commuting} measurements from our non-signaling theories
that are present in quantum theory.
\end{abstract}

\section{Introduction}
In any physical theory, we may consider measurements $M$ that when applied to a state $\rho$ result in some measurement outcome
$k$ with probability $P(k|M)$, depending on $\rho$.
A crucial element in characterizing the power of any physical theory lies in understanding
what probability distributions are indeed possible.
Quantum theory, for example, imposes strict limits on such distributions, which 
greatly affects our ability to perform information processing tasks~\cite{wim:nonlocal}.
One of these limitations is commonly known as an \emph{uncertainty relation}.
We may for example ask whether for some fixed choice of measurements $M_1$ and $M_2$ there even exists
any state such that both distributions can be arbitrarily well defined. That is, is it possible that there exist outcomes $k_1$ and
$k_2$ such that $P(k_1|M_1) = P(k_2|M_2) = 1$? Curiously, it turns out that in quantum theory there do indeed exist pairs of measurements
$M_1$ and $M_2$ for which this is impossible.
Another limitation is known as the \emph{strength of non-local correlations}, which are restrictions on 
the joint probability distributions we can obtain when performing measurements on spatially separated systems.
Classically, these limitations
are known as Bell inequalities, and the 
corresponding limitations in the quantum case are referred to as Tsirelson bounds. 

Since quantum mechanics imposes very stringent restrictions on the possible distributions~\cite{qmp}, we would much like to understand 
their extent and implications. To this end, it is instructive to remove some of these restrictions and investigate how our ability to perform
information processing tasks changes as a result. In this work, we will relax an uncertainty relation, which greatly affects our ability to solve
communication and coding tasks. We will also see that the different kinds of restrictions are very closely related and
show that for example Tsirelson's bound for the CHSH inequality is a consequence of the uncertainty relation of~\cite{ww:cliffordUR}.

\subsection{Previous work}
Previous work has focused on investigating one particular restriction imposed by quantum mechanics, namely its limits on non-local correlations.
Indeed, the existence of non-local correlations in quantum mechanics that are stronger than those allowed by local 
realism~\cite{bell:epr}, but yet strictly weaker than those consistent with the no-signaling principle~\cite{popescu:nonlocal} poses an enigma to the understanding of the foundations of quantum physics. 
What are the properties of quantum mechanics that disallow these stronger correlations~\cite{hardy:axioms}?
And, what possibilities would be opened by the existence of these correlations?
Much of the work exploring these questions has focused on the ``box paradigm'' that was initially inspired by the CHSH inequality~\cite{chsh:nonlocal}. This particular Bell inequality~\cite{bell:epr} can be cast into a form of a simple game between two players, 
Alice and Bob. When the game starts, Alice and Bob are presented with randomly and independently chosen questions $s \in \01$ and $t \in \01$ respectively. They win if and only if they manage to return 
answers $a \in \01$ and $b \in \01$ such that $s \cdot t = a \oplus  b$. Alice and Bob may thereby agree on any strategy before
the game starts, but may not communicate afterwards. Classically, that is in any model based on local realism, 
this strategy consists of shared randomness. It has been shown~\cite{chsh:nonlocal} that for any such strategy we have
$$
\gamma \assign \frac{1}{4} \sum_{s,t \in \01} \Pr[s \cdot t = a_s \oplus b_t] \leq \frac{3}{4},
$$
where $\Pr[s \cdot t = a_s \oplus b_t]$ is the probability that Alice and Bob return winning answers $a_s$ and $b_t$ when presented
with questions $s$ and $t$.
Quantumly, Alice and Bob may choose any shared quantum state together
with local measurements as part of their strategy. This allows them to violate the inequality above, but curiously
only up to a value 
$$
\gamma \leq \frac{1}{2} + \frac{1}{2\sqrt{2}},
$$
known as Tsirelson's bound~\cite{tsirel:original,tsirel:separated}. We will see later that there exists a state $\ket{\Psi}_{AB}$ shared
by Alice and Bob that achieves this bound when Alice and Bob perform measurements given by the observables
$A_0 = B_0 = X$ and $A_1 = B_1 = Z$ where we use $A_s$ and $B_t$ to denote the measurement corresponding to
questions $s$ and $t$ respectively.
The non-signaling principle that disallows faster than light communication between Alice and Bob alone
does not impose such a restrictive bound.
Hence, Popescu and
Rohrlich~\cite{popescu:nonlocal,popescu:nonlocal2,popescu:nonlocal3}
raised the question why nature is not more 'non-local'? That is,
why does quantum mechanics not allow for a stronger violation of
the CHSH inequality up to the maximal value of 1? To gain more
insight into this question, they constructed a toy-theory based on so-called
PR-boxes~\cite{PRbox}. Each such box takes inputs $s,t \in \01$ from
Alice and Bob respectively and simply outputs randomly chosen measurement outcomes
$a_s$,$b_t$ such that $s\cdot t = a_s \oplus b_t$. Each such box can be used exactly once, and no
notion of post-measurement states exists. Note that Alice
and Bob still cannot use this box to transmit any information.
However, since we have for all $s$ and $t$ that
$\Pr[s \cdot t = a_s \oplus b_t] = 1$, Tsirelson's bound is clearly is violated.
It is interesting to consider how our ability to perform information processing tasks changes, 
if PR-boxes indeed existed.
For example, it has been shown that Alice and Bob can use such PR-boxes to compute any Boolean 
function $f:\01^{2n} \rightarrow \01$ of their individual inputs $x \in \01^n$ and $y \in \01^n$ by communicating 
only a \emph{single} bit~\cite{wim:nonlocal}, which is even true when the boxes have slight 
imperfections~\cite{falk:nonlocal}.

Much interest has since been devoted to the study of such PR-boxes and their generalizations known as non-local boxes~\cite{wolf:universal,wehner05b,dariano,masanes:nonsignaling,barnum:broadcasting,barnum:teleport,imai}.
In particular, they have been 
incorporated in a very nice way 
into generalized non-signaling theories (GNST) due to Barrett~\cite{barrett:nonlocal} 
(the relation of such theories to generalizations of quantum theory is due to Hardy~\cite{hardy:axioms})
as a means of exploring foundational questions in quantum information. Intuitively, such theories allow for ``boxes'' involving many more inputs for
one or more players/systems, and also allow for some transformations between such boxes.
Both theories seek out physically motivated properties that single out 
quantum mechanics from other theories such as the classical world. 
These theories have also found interesting applications in deriving 
new bounds for quantum mechanics itself, e.g., monogamy of entanglement~\cite{toner:mono1}. 

In such a theory, $n$-partite
states are characterized by the probabilities of obtaining certain outcomes when performing a fixed set of local fiducial measurements 
on each system. For example, to describe a non-local box, consider a bipartite system, where Alice holds the first and Bob the
second system. We will label both Alice and Bob's measurements using $X$ and $Z$ in analogy to the quantum setting.
For convenience we will also label the outcomes using $a,b \in \01$, where the actual outcomes of $X$ and $Z$ in the quantum
setting could be recovered as $(-1)^a$, and use $p(A|M)$ to denote the probability of obtaining outcomes $A$ for measurements $M$. 
A non-local box is now given by the probabilities 
$p(0,0|X,X) = p(0,0|X,Z) = p(0,0|Z,X) = 1/2$,
$p(1,1|X,X) = p(1,1|X,Z) = p(1,1|Z,X) = 1/2$,
$p(0,1|Z,Z) = p(1,0|Z,Z) = 1/2$ and $p(A|M) = 0$ otherwise.
We will describe such theories
in more detail in section~\ref{sec:gnst}. We will also refer to GNST using the commonly used term ``box-world''. 

\subsection{Relaxed uncertainty relations}

Even when allowing more than two measurements and outcomes, such boxes remain very artificial constructs and it is not quite clear how they relate to quantum theory. 
In this note, we hope to provide a more intuitive understanding by 
showing that superstrong correlations can indeed be obtained by relaxing an uncertainty relation known to hold
in quantum theory. Consider \emph{any} anti-commuting observables $\Gamma_1,\ldots,\Gamma_{2n}$ satisfying
$$
\{\Gamma_j, \Gamma_k\} = 0
$$
whenever $j \neq k$ and
$$
\Gamma_j^2 = \id,
$$
for any $j \in [2n]$, and let $\Gamma_0 = i\Gamma_1\ldots\Gamma_{2n}$ (see section~\ref{sec:prelim} on how to construct such operators). 
It was shown in~\cite{ww:cliffordUR} that any quantum state obeys
\begin{equation}\label{eq:cliffordUR}
\sum_{j = 0}^{2n} \Tr\left(\Gamma_j \rho\right)^2 \leq 1,
\end{equation}
which also lead to several entropic uncertainty relations for such observables. 
To see why Eq.~(\ref{eq:cliffordUR}) itself can be understood as an uncertainty relation note
that $\Tr(\Gamma_j \rho)$ is the expectation value of measuring the observable $\Gamma_j$
on $\rho$. The probability of obtaining a measurement outcome $b \in \{\pm 1\}$ can furthermore
be written as $p(b|\Gamma_j) = 1/2 + b \Tr(\Gamma_j \rho)/2$. Hence, $\Tr(\Gamma_j \rho)$ can also be understood
as the bias towards a particular measurement outcome. Eq.~(\ref{eq:cliffordUR}) now tells us that this bias cannot
be arbitrarily large for all measurements $\Gamma_j$.
Note that we could rewrite the condition of Eq.~(\ref{eq:cliffordUR}) as
$||v||^2_2 \leq 1$
where $v = (\Tr(\Gamma_1\rho),\ldots,\Tr(\Gamma_{2n}\rho))$.
Whereas the uncertainty relations of~\cite{ww:cliffordUR} may appear unrelated to the problem of determining the strength of non-local correlations,
we will see later that Tsirelson's bound for the CHSH inequality is in fact a consequence of Eq.~(\ref{eq:cliffordUR}), 
when we use the fact that local anti-commutation and maximal violations of the CHSH inequality are closely 
related~\cite{tsirel:original, toner:monogamy,uffink:antiComm}.
Thus, as one might intuitively guess, bounds for the strength of non-local correlations are indeed closely related to uncertainty
relations, and such connections have been observed in a different form by~\cite{masanes:nonsignaling, barrett:nonlocal}.

What happens if we merely ask for $||v||^p_p \leq 1$, where $||\cdot||_p$ is the $p$-norm of the vector
$v$? Since Eq.~(\ref{eq:cliffordUR}) must hold for any quantum state, that is for any positive semi-definite matrix $\rho$
with $\Tr(\rho)=1$, it is clear that this allows operators $\rho$ which are no longer positive semi-definite. In the spirit
of Barrett's GNST, we will however restrict ourselves to allowing a particular set of fiducial measurements only, for which
the probabilities will remain positive and thus well-defined. In section~\ref{sec:pnonlocal}, we will describe a hierarchy
of such ``theories'' in detail, and investigate their power with respect to non-local correlations and information processing
problems. In particular, we will see that 
\begin{itemize}
\item For the CHSH inequality, we can obtain at most 
$$
\gamma = \frac{1}{2} + \frac{1}{2(2)^{1/p}} \mbox{ for } ||v||^p_p \leq 1.
$$
where in the limit of $p \rightarrow \infty$ the right-hand side becomes $1$, and we have a state that acts analogous
to a non-local box.
\item Furthermore, any unique XOR-game can be played with perfect success for $p \rightarrow \infty$.
\end{itemize}
It is instructive to consider what our relaxed uncertainty relation means in the case of a single qubit.
Note that for quantum mechanics we have $p=2$ in which case Eq.~(\ref{eq:cliffordUR}) corresponds to the statement that $v$ must lie inside the Bloch sphere.
Allowing different values of $p$ now constraints us to the corresponding $p$-spheres as depicted in Figure~\ref{fig:psphere}.
\begin{figure}\label{fig:psphere}
\begin{center}
\includegraphics[scale=0.7]{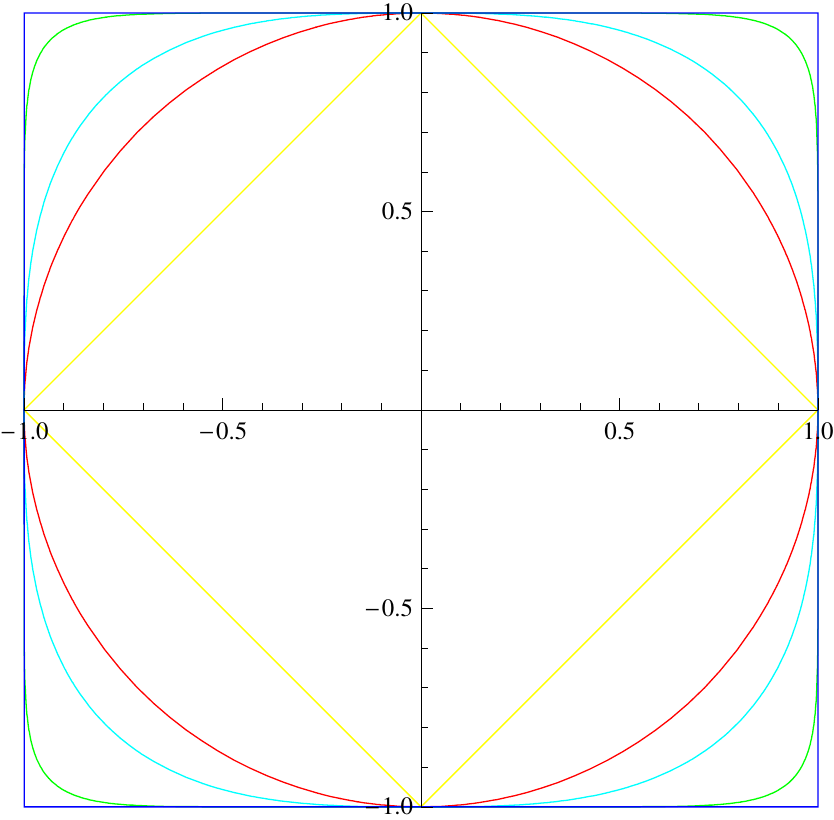}
\caption{$p$-norm unit circles in dimension 2 for $p=1,2,3,10,10000$}
\end{center}
\end{figure}
It is interesting to consider that even though for $p > 2$ we obtain non-local correlations that are \emph{stronger} than what quantum theory allows, we now
have a \emph{weaker} uncertainty relation than in quantum theory. It has previously been noted by Barrett~\cite{barrett:nonlocal} that GNST has no uncertainty relations for particular measurements. Our work makes this relation very intuitive.
In particular, for the case of $p\rightarrow \infty$ corresponding to a non-local box
we essentially place no restrictions on the bias $\Tr(\Gamma_j \rho)$ at all. Since Eq.~(\ref{eq:cliffordUR}) leads to the entropic uncertainty
relations on which the security of the protocols in the bounded-quantum-storage model~\cite{serge:bounded,serge:new,WW08:compose} is based, it may be worth
considering how certain cryptographic tasks change in the setting of non-local boxes. Indeed, it has recently been shown~\cite{wolf:personal} that privacy
amplification fails in a world based on non-local boxes. Whereas it is known that cryptographic tasks such as bit commitment and oblivious transfer
are compatible with the no-signaling principle~\cite{wehner05b}, little is known about them in general theories~\cite{ben:crypto}. 

It should be noted that except for a single qubit, Eq.~(\ref{eq:cliffordUR}) is of course only a necessary and not a sufficient condition for
$\rho \geq 0$. In higher dimensions, such relations are much more involved, but have been obtained for certain operators~\cite{kimura,bloch2,dietz:bloch} 
and also some operators relating more closely to unbiased measurements~\cite{wehner:bloch}. Relaxing this particular uncertainty relation is thus
only one way to go. Yet, due to the rich structure of the Clifford algebra of operators $\Gamma_1,\ldots,\Gamma_{2n}$ 
and their central importance for entropic uncertainty relations and so-called XOR non-local games (also known as two-party correlation inequalities) with
2 measurement outcomes, this small relaxation allows us to gain some insights into their role in quantum information processing tasks.

\subsection{Information processing in generalized non-local theories}

Inspired by these relaxations in terms of an operator $\rho$, we then construct a hierarchy of $p$-GNST theories exhibiting similar constraints. 
For such theories, we identify a single gbit (defined in~\cite{barrett:nonlocal}) with a single qubit obeying the relaxed uncertainty relations above.
That is, we will think of a single gbit as allowing three fiducial measurements labeled $X$, $Z$ and $Y$ in analogy to the quantum case. Whereas 
this choice is of course again quite arbitrary,
and heavily inspired by the quantum setting, it will allow us to gain a slightly better understanding of the relation of ``box-world'' and quantum theory later on.
We show that the states we allow above, as well as states in $p$-GNST's have several properties that set them apart from quantum theory. 
In particular, we will see that
\begin{itemize}
\item In $p$-GNST, there exists superstrong random access encodings. For example, there exists an encoding of $N=3^n$ bits into
$(2n+1)^{3/p} n$ \emph{gbits} such that we can retrieve any bit with probability $1-\eps$ for $\eps=2 \exp(-(2n+1)^{1/p}/2)$.
Quantumly on the other hand it is known that we require at least $(1-h(1-\eps))N$ \emph{qubits} to encode $N$ classical bits with the same recovery probability, where $h$ denotes the binary Shannon entropy.
\item As a consequence, in $p$-GNST there exist single server PIR scheme with $O(\polylog(N))$ bits of communication for an $N$ bit
database with large $N$, whereas quantumly $\Omega(N)$ bits are needed.
\item On the other hand, we show that in GNST it becomes much harder to learn a state in the sense of~\cite{aaronson:learn}.
In fact, unlike in the quantum setting, we can essentially not ignore even a small part of the information we are given about
a state.
\end{itemize}
Note that we thereby compare \emph{units of information}, gbits vs. qubits, irrespective of a physical dimension, where gbits were previously
defined in~\cite{barrett:nonlocal}.
It may not be surprising that such effects exist for Hermitian operators $\rho$, when all we essentially demand is
that the condition $||v||_p^p \leq 1$ is obeyed for any set of anti-commuting measurements. However, it will be interesting
to consider why for example the superstrong random access code encodings we find above are disallowed in quantum theory, but
allowed in GNST.

\subsection{Commuting measurements}

Although the results of local measurements suffice to describe quantum states \cite{hardy:axioms}, our results 
suggest that building a toy-theory around local measurements acting on fixed systems alone (such as GNST) may miss part 
of the flavor when considering some applications. Quantum mechanics has a rich structure of commuting and anti-commuting measurements built in which make 
no particular reference to locality. Uncertainty relations impose restrictions for non-commuting measurements, such as
for example the anti-commuting measurements $\Gamma_1, \ldots, \Gamma_{2n}$. However, we will see in 
section~\ref{sec:consistency} that also certain sets of commuting measurements cannot have arbitrary expectation values when measured
on a particular state $\rho$. 
As a simple example, 
consider a 2 qubit system shared between Alice and Bob, and consider the measurement $X \otimes \id$, $\id \otimes X$ and $X \otimes X$. 
Suppose that we have $\Tr((X \otimes \id)\rho) = \Tr((\id \otimes X)\rho) = 1$. This tells us
that when Alice and Bob measure $X$ locally, they obtain an outcome of `1' each with probability 1. However,
the measurement of $X \otimes X$ can very intuitively be viewed as Alice and Bob performing a local measurement of $X$ and taking the product
of their outcomes. Hence, we do not expect a simultaneous assignment of $\Tr((X \otimes X)\rho) = -1$ to be consistent 
with the previous two expectation values. We will formalize this intuition in section~\ref{sec:consistency}, where
we will derive a series of conditions such expectation values must obey which in spirit is similar to~\cite{qmp}.

GNST does satisfy these conditions for measurements that commute because they act on different subsystems. 
It does not exhibit any inconsistencies otherwise, as no commutation relations are defined for measurements on the same
system. The issue of such inconsistencies is further circumvented by the simple fact that a non-local box can only
be used once, and there is no notion of subsequent measurements on the same system.
This of course is perfectly adequate for studying the strength of non-local correlation between two space-like
separated systems for example, and led to such perplexing results as~\cite{wim:nonlocal}.
We will however see that it is essentially this lack of additional constraints that allows us to form
superstrong random access codes for example, and may indicate that using ``box-world'' to investigate
the role of the strength of non-local correlations within quantum theory itself is possibly doomed to fail.
It also indicates why defining a consistent notion of 'post-measurement' states for non-local boxes
is quite difficult, since many constraints that would allow such a task to succeed are simply not
present in box-world.

To see how box-world differs from quantum theory consider the measurements
$M_1 = X \otimes Z$, $M_2 = Z \otimes X$ and $M_3 = - XZ \otimes XZ$. These are related in exactly the same way 
as the measurements we considered above, except that in GNST there is no notion that $M_1$ and $M_2$ commute. 
Yet, we intuitively expect similar conditions to hold as for the measurements above when trying to form an analogy to the quantum setting. 
Indeed, one can 
easily
construct a unitary transformation that maps the 
measurements $M_1, M_2$ and $M_3$ into a form analogous to the above, where two of the measurements act on different systems.~\footnote{Consider $U = (\id \otimes H)\mbox{CNOT}(\id \otimes H)$}
In GNST, however, the separation into different systems is always a given, which may lead to difficulties when examining some problems
which are not really concerned with correlations among two distant systems alone, but to information processing in general.

\subsection{Outline}
Whereas we only examine a very small piece of the puzzle, our work hopes to shed some light on the relation between uncertainty relations,
non-local correlations and the role of above mentioned consistency constraints in information processing.
In section~\ref{sec:prelim} we first explain the basic concepts we need to refer to.
commuting measurements in more detail. In section~\ref{sec:ptheories} we then define a range of simple ``theories'' obtained by
relaxing the uncertainty relation for anti-commuting observables. To highlight the analogy with non-local boxes, we then define
a range of similar GNST-like theories in section~\ref{sec:gnst}. In sections~\ref{sec:superstrong}, \ref{sec:rac}, and \ref{sec:info} we then investigate the power of such
theories with respect to non-local correlations, random access codes, and information processing problems respectively. 
In section~\ref{sec:constraints} we then investigate why such effects are possible within GNST, but not in quantum theory.
Table \ref{table:results} summarizes similarities and differences among theories.

\section{Preliminaries}\label{sec:prelim}

\subsection{Basic concepts}\label{sec:basics}

In the following, we write $[n] \assign \{1,\ldots,n\}$ and use $X$, $Z$ and $Y$ to denote the well-known Pauli matrices~\cite{nielsen&chuang:qc}.
We also speak of a \emph{string of Paulis} to refer to a matrix of the form
\be\label{eq:pauligroup}
S_{ab} \assign X^{a_1}Z^{b_1} \otimes \ldots \otimes X^{a_n}Z^{b_n},
\ee
with $a = (a_1,\ldots,a_n)$, $b = (b_1,\ldots,b_n)$ and $a_j,b_j \in \01$.
We sometimes write the Pauli operator acting on subsystem $j$, with identity on the other subsystems as $$X_j = \id^{\otimes j-1} \otimes X \otimes  \id^{\otimes n-j-1} $$

The \emph{Pauli basis expansion} of a density matrix $\rho$ is given
by $\rho = (\id + \sum_{a,b} s_{ab} S_{ab})/d$, where we call
$s_{ab}$ the \emph{coefficient} of $S_{ab}$.
Consider the form $f(a,b,a',b') = (a,b') + (a',b)$, where we
write $(a,b) = \sum_j a_j b_j \mod 2$.
It it straightforward to convince yourself that for any pair $S_{ab}$
and $S_{a'b'}$ either $[S_{ab},S_{a'b'}] = 0$ if $f(a,b,a',b') = 0$ or $\{S_{ab},S_{a'b'}\} = 0$ if $f(a,b,a',b') = 1$.
Whereas Eq.~(\ref{eq:cliffordUR}) holds for any choice of anti-commuting measurements, it is worth noting that
in dimension $d=2^n$ we can find at most $2n+1$ anti-commuting operators given by
\begin{eqnarray*}
\Gamma_{2j-1} &=& Y^{\otimes (j-1)} \otimes X \otimes \id^{\otimes (n-j)}\\
\Gamma_{2j} &=& Y^{\otimes (j-1)} \otimes Z \otimes \id^{\otimes (n-j)},
\end{eqnarray*}
for $j=1,\ldots,n$ and $\Gamma_0 = i \Gamma_1\ldots\Gamma_{2n}$. Note that for $n=1$ we have
$\Gamma_1 = X$, $\Gamma_2 = Z$, $\Gamma_0 = Y$ and Eq.~(\ref{eq:cliffordUR}) is equivalent to the Bloch sphere 
condition.
We will also need the notion of a $p$-norm of a vector $v = (v_1,\ldots,v_n) \in \Real^n$ which is defined
as
$$
||v||_p \assign \left(\sum_{j=1}^n |v_j|^p\right)^{1/p}.
$$
Note that for $p=2$ this is just the Euclidean norm. Of particular interest to us will also be the $\infty$-norm defined
as $||v||_\infty \assign \lim_{p \rightarrow \infty} ||v||_p$ which can also be written as
$$
||v||_\infty = \max(|v_1|,\ldots,|v_n|).
$$

\subsection{Probability distributions}\label{probability}

Unlike previous descriptions of general probabilistic theories, our notation must be versatile enough to accommodate arbitrary choices of simultaneous commuting measurements, even if they do not act on separate subsystems. In quantum mechanics we may choose to measure $X \otimes X$ along with either $X \otimes \id, \id \otimes X$, or $Z \otimes Z,XZ \otimes XZ$. We will see that including this flexibility in a more general theory leads to new constraints.

First, we want to consider some finite set of measurements $\setM = \{M_1,\ldots,M_N\}$ where without loss of generality we assume that each measurement has the same finite set of outcomes $\setA$ and the $\setM$ is ordered lexiocraphically.
Although we initially impose no structure on $\setM$, in analogy to quantum mechanics we consider certain collections of measurements $\CM \subseteq \setM$ to have some property which directly corresponds to simultaneous measurability. In particular, we will consider the set of possible experiments 
$$\setE \assign \{ \CM \subseteq \setM ~\wedge~ \forall M_i,M_j \in \CM \mbox{  sim}(M_i,M_j)=0 \},$$ 
where ``sim'' is a predicate indicating simultaneous measurability that remains to be specified.
Of particular concern to us will be
the probability distributions $p$ over the outcomes $A \in \setAC$ of some set of simultaneously performed measurements $\CM \in \setE$.
We use $p(A|\CM)$ to denote the probability of obtaining outcomes $A = (A_1, A_2,\ldots,A_{|C|} ) \in \setAC$ for measurements $\CM \subseteq \setM$ where
we wlog take $\CM$ to be ordered lexicographically.
For simplicity, we will also write $p(A_1,\ldots,A_{n}|M_1,\ldots,M_n) \assign p((A_1,\ldots,A_n)|\{M_1,\ldots,M_n\})$.

What conditions do the functions $p: \setAC \times \CM \rightarrow [0,1]$ have to fulfill be a valid probability distribution for any experiment $\CM \in \setE$?
We require that the following conditions need to be satisfied for \emph{any} probability distribution
\begin{enumerate}
\item[(1)] Normalization: 
$\forall \CM \in \setE, \sum_{A \in \setAC} p(A|\CM) = 1$.
\item[(2)] Positivity: $\forall \CM \in \setE, \forall A \in \setAC, p(A|\CM) \geq 0$.
\end{enumerate}
The next condition may appear unfamiliar at first glance. 
Intuitively it says that the distributions of outcomes we obtain for commuting measurements are independent of what other commuting measurements
we perform.
\begin{enumerate}
\item[(3)] Independence: 
$$
\forall \CM, \CM' \in \setE \mbox{ with } \CM \subseteq \CM',~~ p((A_1,\ldots,A_{|\CM|})|\CM) = \sum_{A_{|\CM|+1},\ldots,A_{|\CM'|} \in \setA^{\times |C'|}} p((A_1,\ldots,A_{|\CM'|})|\CM'),
$$
\end{enumerate}
where, without loss of generality, we take the first $|\CM|$ outcomes to be associated with the measurements in $\CM$. 

Throughout this text, we explore the result of choosing two different ways of choosing simultaneous measurements. 
First, we consider simultaneous measurements on distinct systems as reflected in the construction of non-local boxes. 
Second, we consider a more general notion of such measurements based on commutation relations as in quantum mechanics. 
Note that in the quantum case such sets of mutually commuting measurements induce a partitioning of the Hilbert space into different systems in the finite-dimensional setting~\cite{summers:qftIndep, qmp}. 

Consider the set of measurements $\setM_P$ to be strings of Paulis on $n$-partite systems as defined in section~\ref{sec:basics}. The two different notions
of simultaneous measurements can now be expressed in two different choices of $\mbox{sim}(M_i,M_j)$, leading to two different sets of realizable experiments. 
To capture the first notion, we let
$$\setE_L \assign \{ \CM \subseteq \setM_P ~\wedge~ \forall M_i,M_j \in \CM \mbox{  local}(M_i,M_j)=0 \},$$
 where local$(M_i,M_j) = 0$ if and only if $M_i$ and $M_j$ act on different subsystems. For example, we have local$(X\otimes \id,\id \otimes Z)=0$. 
Second, we let
$$\setE_C \assign \{ \CM \subseteq \setM_P ~\wedge~ \forall M_i,M_j \in \CM ~[M_i,M_j]=0 \},$$ 
where all commuting measurements are simultaneously observable, as in quantum mechanics.
Clearly, $\setE_L \subseteq \setE_C$, since two measurements acting on two different subsystems commute.

When we restrict ourselves to $\setE_L$ we can express the independence condition from above in the more familiar form of no-signaling:
\begin{enumerate}
\item[(3')] No-signaling: $$\forall \CM,\CM' \in \setE_L \mbox{ with }  \CM \subseteq \CM',~~ p((A_1,\ldots,A_{|\CM|})|\CM) = \sum_{A_{|\CM|+1},\ldots,A_{|\CM'|} \in 
\setA^{\times |C'|}} p((A_1,\ldots,A_{|\CM'|})|\CM').$$
\end{enumerate}
Intuitively, the no-signaling condition just dictates that 
the marginal distribution of a particular subset of systems is \emph{independent} of the 
measurement choices on a disjoint subset of systems.
Therefore, we can simplify our description of marginals of no-signaling distributions to just $p ( A \in \setAC | \CM') = p(A | \CM)$, 
where the measurement choices on other parties are arbitrary. We will later see that imposing only the special case of the no-signaling
condition, versus the full independence condition of (3), makes a crucial difference in the power of the resulting theory with respect
to encoding information.

\begin{examp}\label{nlbox} 
Consider the set of local experiments for two parties with $\setA =\{-1,1\}, \setM =\{X_1, Z_1, X_2, Z_2 \}$. Let the probability distribution $p(A|C)$ be described by the following table. 

\begin{center}
\begin{tabular}{c|cccc|c}
${A}$ &&&&& \\ \hline &&&&&\\[-1.5ex]
$(1,1)$ & $\half$ & $\half$& $\half$& $0$& \\[1ex]
$(1,-1)$ &$0$ & $0$&$0$ &$\half$ & \\[1ex]
$(-1,1)$ & $0$&$0$ & $0$& $\half$& \\[1ex]
$(-1,-1)$ &$\half$ & $\half$&$\half$ &$0$ & \\[1ex] \hline
~&$\{X_1,X_2\}$&$\{X_1,Z_2\}$&$\{Z_1,X_2\}$&$\{Z_1,Z_2\}$& ${\CM}$ \\
\end{tabular}
\end{center}

Clearly, we have positivity, and the sum over each measurement setting (column) is $1$. Finally, note that the marginal probability distribution for either party is constant, $\forall \CM \in \setE_L, \forall A_1 \in \setA, \sum_{A_2\in\setA} p ((A_1,A_2) | \CM) = \half$, therefore this distribution is no-signaling.
\end{examp}

\subsection{Moments}\label{sec:moments}
Any finite, discrete probability distribution has a dual representation in terms of a finite number of moments~\cite{wainwrightjordan}. 
We define the product of the outcomes $A = (A_1,\ldots,A_{|C|}) \in \setAC$ of a collection of measurements $\CM \in \setE$ as $A^* = \prod_{i=1}^{|\CM |} A_i$.
The moment for this measurement is defined as
\begin{equation}\label{eq:momentDef}
m (\CM) \assign \sum_{A \in \setAC} p(A|\CM) A^*.
\end{equation}
Note that for the identity measurement this means $m(\id) = 1$ because of normalization. 
Also, if you consider the moment for some subset of $\CM$, by the independence principle this definition gives a unique value which does not depend on the choice of other measurements made simultaneously.

Since we will only be concerned with measurements with two outcomes $\setA=\{\pm 1\}$, we now restrict
ourselves to this case for simplicity. For the measurement of a single observable $\CM=\{M_1\}$ with outcome
$A_1 \in \setA$, we can easily recover the probabilities from the moments as
\begin{equation}\label{eq:probDef}
p((A_1) | \{M_1\}) = \frac{1}{2}\left(1 + A_1 ~ m(\{M_1\})\right).
\end{equation}
In subsequent notation, we will drop the brackets within parentheses when it increases readability.

Note that we can recover the probability for a specific set of outcomes $\hat{A} \in \setAC$ and measurements $\CM \in \setE$ from these moments. Without loss of generality, let $\CM = \{M_1,\ldots,M_n\}$.
\begin{eqnarray*}
&&\frac{1}{2^n} \sum_{\CM' \subseteq \CM} m(\CM')~\prod_{i,M_i\in \CM'} \hat{A}_i\\
&=& \frac{1}{2^n} \sum_{\CM' \subseteq \CM} \left( \sum_{A\in \setA^{\times |\CM'|}} p(A|\CM') ~\prod_{i,M_i\in \CM'} A_i \right)~\prod_{i,M_i\in \CM'} \hat{A}_i\\
&=& \frac{1}{2^n}  \sum_{A\in \setAC} p(A|\CM)~ \sum_{\CM' \subseteq \CM} ~\prod_{i,M_i\in \CM'} A_i \hat{A}_i
\end{eqnarray*}
The second line simply uses the definition of $m(\CM')$ and the third line uses the independence principle to write $p(A|\CM')$ in terms of $p(A|\CM)$, allowing us to move the sum over $\CM'$ inside.
Now note that the sum over $\CM'$ can be broken into $n$ sums over whether or not $M_i \in \CM'$. For each $M_i$, if it is in $\CM'$ we get a factor of $A_i \hat{A}_i$, otherwise a factor of $1$. 
\begin{eqnarray*}
&=& \frac{1}{2^n}  \sum_{A\in \setAC} p(A|\CM)~ \prod_{i=1}^n (1+A_i \hat{A}_i) \\
\end{eqnarray*}
Because the outcomes can only be $\pm 1$, the sum can give us only $0$ or $2$.
\begin{eqnarray*}
&=& \frac{1}{2^n}  \sum_{A\in \setAC} p(A|\CM)~ \prod_{i=1}^n 2 \delta_{A_i, \hat{A}_i} \\
&=& \frac{1}{2^n}  \sum_{A\in \setAC} p(A|\CM)~ 2^n \delta_{A, \hat{A}} \\
&=& p(\hat{A} |\CM)
\end{eqnarray*}

\subsection{Consistency constraints}\label{sec:ks}\label{sec:constraints}\label{sec:consistency}

We are now ready to investigate the constraints that arise due to simultaneous measurement of commuting observables and that will play a crucial role in 
understanding the differences between quantum theory and $p$-GNST.
Imagine two commuting measurements $[M_i,M_j]=0$, and their product $M_k=M_i M_j$. 
In quantum mechanics the outcome of the measurement $M_k$ is the same as the product of the outcomes of $M_i$ and $M_j$,
 which can be verified by expanding
$M_k$ in terms of $M_i$ and $M_j$ and using the fact that they have a joint eigenbasis. What happens if we take this to be true in any theory?
If we are only allowed to make local measurements, then this is a moot point. We can only get $X\otimes X$ by measuring $X\otimes \id$ and $\id \otimes X$ and multiplying the results.

But if we are allowed to make any combination of commuting measurements, this will impose some interesting conditions.
For example, in the quantum 
case we may have $M_1 = X \otimes X$, $M_2 = Z \otimes Z$ and $M_3 = XZ \otimes XZ$.
To see that this has consequences in terms of the moments, consider the simple example where
$m(M_1) = 1$ and $m(M_2) = 1$, which means that we will deterministically observe outcomes $A(M_1) = A(M_2) = 1$. Hence, $m(M_3) = -1$
should intuitively not be compatible with these two moments for $M_1$ and $M_2$. 

How can we formalize these conditions?
For example, Eq.~(\ref{eq:momentDef}) gives us that $$m(M_1 M_2) = m(M_1,M_2),$$ if we insist that outcomes of products of measurements equal the product of outcomes of individual measurements.
For a given set of commuting measurements $\CM= \{M_1,\ldots,M_m\}$ with $M_j^2 = \id$, let $s(M)$ be the $2^m$ element vector
whose $k$-th entry is given by 
\begin{equation}\label{eq:sVector}
s \assign [s(\CM)]_k \assign M_1^{k_1}M_2^{k_2}\ldots M_m^{k_m},
\end{equation}
with $k \in \01^m$ in lexicographic order.
We now define the \emph{moment matrix} $K_s$
by letting the entry in the $i$-row and $j$-th column be given by
$$
[K_s]_{ij} \assign m(s_i s_j)/2^m.
$$

\begin{claim}[Adapted from Wainwright and Jordan \cite{wainwrightjordan}]
Let $\CM = \{M_1,\ldots,M_m\}$ be a set of commuting measurements. 
Then $K_s \geq 0$ if and only if $p$ is a probability distribution (satisfying constraints
(1) and (2)).
\end{claim}
\begin{proof}
In addition to $K_s$, we define two more $2^{m} \times 2^{m}$ matrices, whose components are labeled by vectors $i,j \in \{0,1 \}^{m}$ in lexicographic order
as
\benn
\ [P]_{ij} &= \delta_{ij} p( A = ((-1)^{i_1},\ldots, (-1)^{i_{m}}) | \CM).\\
[B]_{ij} &= \frac{1}{2^{m/2}}(-1)^{i \cdot j},
\eenn
It is easily verified that $B$ is a unitary matrix.
Note that $B$ is an example of a Hadamard matrix.
Now we will show that $K_s = B P B^{\top}$.
\benn
\left[ B P B^{\top} \right]_{ij} &= \frac1{2^m} \sum_{k,l \in \{0,1\}^{m}} (-1)^{i\cdot k} ~ \delta_{kl} ~p(( (-1)^{k_1},\ldots, (-1)^{k_{m}} ) | \CM)  (-1)^{l\cdot j} \\
 &=\frac1{2^m}  \sum_{k \in \{0,1\}^{m}} (-1)^{k \cdot (i\oplus j)}  p(( (-1)^{k_1},\ldots, (-1)^{k_{m}} ) | \CM)\\
 &= \frac1{2^m} \sum_{k \in \{0,1\}^{m}} \prod_{t=1}^{m} ((-1)^{k_t})^{(i_t \oplus j_t)} p(( (-1)^{k_1},\ldots, (-1)^{k_{m}} ) | \CM)\\
  &=\frac1{2^m}  \sum_{A \in \setAC} \prod_{t=1}^{m} A_t^{i_t } A_t^{ j_t}  p( A | \CM)\\
  &=\frac1{2^m}  m(s_i s_j) =  [K_s]_{ij}
\eenn
Clearly, if the probabilities $p(A | \CM)$ are non-negative (2), then $P \geq 0$ if and only if $K\geq 0$ since $B$ is unitary.
Similarly, the fact that $m(\id) = 1$, $B$ is unitary and the trace is cyclic ensures that $p$ satisfies condition (1).
\end{proof}

\begin{examp}\label{nlmoments}
As an example, consider the case of two commuting measurement $M_1$ and $M_2$ with $M_3 = M_1M_2$. We have $s = (\id, M_1, M_2, M_3)$
and
\benn
\mathbf{K}_{s}= \begin{pmatrix}
m(\id) & m(M_1) & m(M_2) & m(M_1 M_2) \\
m(M_1) & m(\id) & m(M_3) & m(M_2) \\
m(M_2) & m(M_3) & m(\id) & m(M_1) \\
m(M_3) & m(M_3) & m(M_1) & m(\id)  \end{pmatrix} 
\equiv \begin{pmatrix} 1&a&b&c \\ a &1&c &b\\b&c&1&a\\c&b&a&1\end{pmatrix}
\eenn
Demanding that the eigenvalues of this matrix, $\lambda = ((1+a-b-c),(-1+a+b-c),(-1+a-b+c),(1+a+b+c))$, be non-negative is enough to ensure that 
$\mathbf{K}_{s} \succeq 0$. Using the Sylvester criteria, we get the alternate constraints that each moment $|a,b,c| \leq 1$ and $1-a^2-b^2-c^2+2 a b c \geq 0$, and $\lambda_1 \lambda_2 \lambda_3 \lambda_4 \geq 0$.
\end{examp}

Our examples are reminiscent of the examples considered in the setting of contextuality~\cite{peres:book}. Note that our constraints are related, but
nevertheless of a different flavor since we only consider such constraints for measurements which all commute. It may be interesting to consider 
such a moment matrix in order to determine how ``non-contextual'' quantum theory is. 
In section~\ref{sec:pgnst} and \ref{sec:pnonlocal} we will develop classes of states which are restricted by imposing specific relationships among various moments.
In particular, it will be of crucial importance whether we merely impose such constraints for measurements acting on different systems, or include
such constraints for all commuting measurements.

\section{$p$-nonlocal theories and their properties}\label{sec:pnonlocal}\label{sec:ptheories}

We now define a series of so-called $p$-nonlocal ``theories'', each one more constrained than the previous. Our definition is thereby motivated by
the uncertainty relations of~\cite{ww:cliffordUR} stated above. 
We later relate our definitions to Barrett's GNST~\cite{barrett:nonlocal} and what are commonly known as non-local boxes. Our aim by constructing this series of simple theories is thereby merely to gain a more intuitive understanding of superstrong
non-local correlations due to non-local boxes.

\subsection{A theory without consistency constraints}

We start with the simplest of all $p$-theories, which forms the basis of all subsequent definitions. In essence, we will simply 
allow states violating the uncertainty relation in~\ref{eq:cliffordUR} without worrying about anything else.
In the spirit of Barrett~\cite{barrett:nonlocal} we start by defining the states which are
allowed in our theory, and then allow all linear transformations preserving the set of allowed states.
For simplicity, we will only consider the case of $d=2^n$.

\begin{definition}
A $d$-dimensional $p$-\emph{bin state} is a $d \times d$ complex Hermitian matrix 
$$
\rho = \frac{1}{d}\left(\id + \sum_{a,b} s_{ab} S_{ab}\right)
$$
satisfying
\begin{enumerate}
\item for all $a,b$, $-1 \leq s_{ab} \leq 1$.
\item for any set of mutually anti-commuting strings of Paulis $A_1,\ldots,A_m \in \Complex^{d \times d}$
$$
\sum_j |\Tr(A_j\rho)|^p \leq 1.
$$
\end{enumerate}
\end{definition}

It remains to be specified what operations and measurements we are allowed to perform on $p$-bin states.
We define
\begin{definition}
A $d$-dimensional $p$-\emph{bin theory} consists of
\begin{enumerate}
\item states $\rho \in \spdBin$ where $\spdBin$ is the set of $d$-dimensional $p$-bin states,
\item linear operations $T: \spdBin \rightarrow \spdBin$,
\item measurements described by observables $S_{ab} = S_{ab}^0 - S_{ab}^1$ where $S_{ab}^0$ and $S_{ab}^1$
are projectors onto the positive and negative eigenspace of $S_{ab}$ respectively.
As in the quantum case we let
$$
p_0 = \Tr(\rho S_{ab}^0) \mbox{ and } p_1 = \Tr(\rho S_{ab}^1).
$$
\end{enumerate}
Starting from a state, we may apply any set of operations $T$ followed by a single measurement.
\end{definition}

Note that by virtue of Eq.~(\ref{eq:cliffordUR}) any quantum state is a $p$-bin state. Note that
the converse however does not hold, since the conditions given above do not imply that a $p$-bin
state $\rho$ is positive semi-definite.
It seems very restrictive to limit ourselves to a single measurement at the end.
The reason for this is that for some $p$, there exist $p$-bin states to start with, valid operations
and measurements, followed by another operation that give us a states that are no longer a $p$-bin 
states~\cite{kitaev:magic}.
We return to this question, when we consider the set of allowed operations below.

Note that the above definition is well-defined. First, we want
that for any measurement $S_{ab}$, $\{p_0,p_1\}$ forms a valid probability distribution. 
A small calculation gives us that 
any $p$-nonlocal state $\rho$ we have
$$
p_v = \Tr(\rho S_{ab}^v) = \frac{1}{2}\left(1 + (-1)^v s_{ab}\right),
$$
and thus $0 \leq p_b \leq 1$ and $p_0 + p_1 = 1$.
Second, we want the non-signaling conditions to hold.
When measuring $S_{ab} \otimes S_{a'b'}$ on a bipartite state 
$$
\rho_{AB} = \frac{1}{d}\left(\id + \sum_{\ell,m,\ell',m'} S_{\ell,m} \otimes S_{\ell',m'}\right)
$$ 
we have that the probability to obtain outcome $u$ for the measurement on the first system is given by
$$
\Pr[u|ab,a'b'] = \sum_{v \in {0,1}} \Tr\left(\rho_{AB} (S_{ab}^u \otimes S_{a'b'}^v)\right) = \frac{1}{2}(\id + (-1)^u s_{a,b,0,0}),
$$
and hence $\Pr[u|ab,a'b'] = \Pr[u|ab,a''b'']$ for all $a',b',a'',b''$ as desired. 
A similar argument can be
made to show that the more general independence condition is satisfied.

\subsubsection{Basic Properties}

We now state some basic properties of this theory, which will also hold for a more restricted $p$-nonlocal theory as outlined below.

\begin{claim}\label{claim:inclusion}
If $\rho$ is a $p$-bin state, then $\rho$ is also a $q$-bin state
for $p,q \in \Integer$ with $q \geq p$.
\end{claim}
\begin{proof}
This follows immediately from the fact that for any $r \in [0,1]$ we have $r^q \leq r^p$.
\end{proof}

Below, we will apply circuits consisting of the Clifford gates $\{CNOT,X,Z,Y,H\}$ and $\id$. It is easy to see that such unitary operations
are allowed transformations taking $p$-bin states to $p$-bin states.

\begin{claim}\label{claim:operations}
Let $\rho \in \spdBin$. Then for any circuit $U$ consisting solely of the gates $\{CNOT,X,Z,Y,H,\id\}$
we have $U \rho U^\dagger \in \spdBin$.
\end{claim}
\begin{proof}
Note that $U$ is composed of single unitaries $U_j = \id^{j-1} \otimes V \otimes \id^{n-j}$
with $V \in \{X,Z,Y,H\}$ and unitaries $U'_j = \id^{j-1} \otimes \mbox{CNOT} \otimes \id^{n-j-1}$.
First, it is straightforward to verify that for any $a,b \in \{0,1\}^n$, there exist $a',b'\in \01^n$ such that
$U_j S_{ab} U_j^\dagger = S_{a'b'}$, and similarly for $U'_j$.
Second, applying a unitary to any set of anti-commuting operators again gives us anti-commuting operators.
Hence, since we have $\sum_j |\Tr(A_j \rho)|^p \leq 1$ for \emph{any} set of anti-commuting strings of Paulis,
the resulting state will also have this property.
\end{proof}

It will also be useful to know that
\begin{claim}\label{claim:tensor}
Let $\rho_1,\ldots, \rho_n \in {\setS_p^2}$. Then $\bigotimes_{i=1}^n \rho_i  \in \setS_p^{2^n}$.
\end{claim}
\begin{proof}
We proceed by induction. By assumption, $\rho_1  \in {\setS_p^2}$. We will show that for any states $\rho  \in {\setS_p^{2^n}}, \sigma \in {\setS_p^2}$, the state $\rho \otimes \sigma \in \setS_p^{2^{n+1}}$.

We need to prove that for any set of mutually anti-commuting Pauli's $A_j \in \Complex^{2^{n+1} \times 2^{n+1}}$
$ \sum_j |\Tr(A_j \rho \otimes \sigma )|^p \leq 1.$ Each $A_j$ can always be written in terms of a Pauli, $B_j$ acting on $\rho$, plus a Pauli $\{\id, X,Y,Z\}$ on $\sigma$. We separate 
the $A_j$ into groups according to which Pauli is appended to $B_j$. Then we can rewrite this as
\benn
\sum_{j_\id} |\Tr((B_{j_{\id}}\otimes \id) (\rho \otimes \sigma) )|^p  + \sum_{j_X} |\Tr((B_{j_X} \otimes X) (\rho \otimes \sigma) )|^p \\ + \sum_{j_Y}|\Tr((B_{j_Y} \otimes Y) (\rho\otimes \sigma) )|^p  + \sum_{j_Z}|\Tr((B_{j_Z} \otimes Z) (\rho\otimes \sigma) )|^p  \\
 = \sum_{j_\id} |\Tr(B_{j_{\id}} \rho )|^p + \sum_{j_X} |\Tr(B_{j_X} \rho )|^p |\Tr(X \sigma )|^p \\+ \sum_{j_Y}|\Tr(B_{j_Y} \rho )|^p |\Tr(Y \sigma )|^p + \sum_{j_Z}|\Tr(B_{j_Z} \rho )|^p |\Tr(Z \sigma )|^p \leq 1
\eenn
Since all the $A_j$ mutually anti-commute, then for different $j, j'$, $\{ B_{j} \otimes X , B_{j'} \otimes X \} = 0$ implies $\{ B_{j}  , B_{j'} \} = 0$, while $\{ B_{j} \otimes X , B_{j'} \otimes Y \} = 0 $ implies $ [ B_{j} , B_{j'} ] = 0$. Then because $\rho \in \setS_p^{2^n}$ and $ \{ B_{j_X}  , B_{j_X'} \} = 0$, and, for similar reasons $ \{ B_{j_X}  , B_{j_{\id}} \} =\{B_{j_{\id}'} , B_{j_{\id}} \} = 0$, we know
$$
\sum_{j_\id} |\Tr(B_{j_{\id}} \rho )|^p + \sum_{j_X} |\Tr(B_{j_X})|^p \leq 1
$$
Now we will shorten our notation by writing
\[
\begin{array}{cc} a_X = |\Tr(X \sigma )|^p & b_X = \sum_{j_X } |\Tr(B_{j_X} \rho)|^p \\ a_Y = |\Tr(Y \sigma )|^p & b_Y = \sum_{j_Y} |\Tr(B_{j_Y} \rho)|^p \\ a_Z=|\Tr(Z \sigma )|^p & b_Z = \sum_{j_Z} |\Tr(B_{j_Z} \rho)|^p \\ & b_\id = \sum_{j_\id} |\Tr(B_{j_{\id}} \rho )|^p\end{array}\]
This allows us to write inequalities implied by the uncertainty relation like:
\begin{eqnarray*}
a_X+a_Y+a_Z \leq 1\\
b_X +b_\id \leq 1 \\
b_Y +b_\id \leq 1 \\
b_Z +b_\id \leq 1
\end{eqnarray*}
We can also see that $a_X,a_Y,a_Z,b_X,b_Y,b_Z,b_\id \geq 0$. The task at hand is to show that these inequalities imply the one required of a state in $\setS_p^{2^{n+1}}$, which we can now rewrite as $$ a_X b_X +a_Y b_Y + a_Z b_Z +b_\id \leq 1.$$
We do this by writing down a sum of products of non-negative quantities like $1- a_X-a_Y-a_Z$ and noting that the result is non-negative.
\benn
a_X (1-b_X - b_\id)+a_Y (1-b_Y - b_\id)+a_Z (1-b_Z - b_\id) + (1-b_\id) ( 1 - a_X -a_Y-a_Z) \geq 0
\eenn
That equation can be rewritten as $ 1- (a_X b_X +a_Y b_Y + a_Z b_Z +b_\id) \geq 0$, which is what we set out to show. Therefore, $\rho \otimes \sigma$ is a valid state, and, by induction, so is $\bigotimes_{i=1}^n \rho_i  \in \setS_p^{2^n}$ for any $n$.
\end{proof}

\subsection{An analogue to box-world}

Note that in the above definition we have not placed any constraints at all on the expectation values of commuting measurements. 
This was not necessary, as we had allowed a single measurement only, where by the above definition $\id \otimes X$ formed
such a single measurement.
Now consider a two-qubit system, i.e., $d = 4$. Suppose that we have for a particular $\rho$ that
$$
\Tr\left((X \otimes \id) \rho\right) = \Tr\left((\id \otimes X) \rho \right) = \Tr\left((X \otimes X) \rho \right) = -1.
$$
Note that $\rho$ can be a perfectly valid state with respect to the definition given above, but yet we would
not consider this to be consistent behavior, if we were allowed to perform subsequent measurements.
We now introduce additional constraints that eliminate this inconsistency. It should be clear from section~\ref{sec:moments}
that that to achieve full consistency we would have to introduce certain constraints for commuting observables in general. Yet, we will
first restrict ourselves to observables on different systems in analogy to ``box-world''. We will show in section~\ref{sec:pgnst} that Barrett's GNST
and non-local boxes essentially correspond to this definition. We will
also see in section~\ref{sec:rac} and~\ref{sec:comm} that these additional constraints play a crucial role in the power of our model
with respect to information processing tasks.

\begin{definition}
A $p$-\emph{box state} is a $p$-bin state $\rho$, where in addition we require that for any set $\CM \in \setE_L$ of measurements
acting on different systems and $s(\CM)$ as defined in Eq.~(\ref{eq:sVector}) we have that the corresponding moment matrix $K_s$ defined in section~\ref{sec:ks} satisfies
$$
K_s \geq 0.
$$
\end{definition}

Note that claims~\ref{claim:inclusion} and~\ref{claim:tensor} holds analogously for $p$-box states. It is important to note though 
that claim~\ref{claim:operations} does not hold in this case, since for example the CNOT operation can lead to states violating the definition.

\subsection{A theory with consistency constraints}

Finally, we will impose all constraints required from our consistency considerations of section~\ref{sec:moments}. 

\begin{definition}
A $p$-\emph{nonlocal state} is a $p$-box state $\rho$, where in addition we require that
for any set of commuting measurements $\CM \in \setE_C$ and $s(\CM)$ as defined in Eq.~(\ref{eq:sVector}) 
we have that the corresponding moment matrix $K_s$ as defined in section~\ref{sec:ks} satisfies
$$
K_s \geq 0.
$$
\end{definition}

Again claims~\ref{claim:inclusion} and~\ref{claim:tensor} hold analogous to the above. When we include all consistency considerations, 
it is also easy to see that claim~\ref{claim:operations} holds for $p$-nonlocal states, since for any allowed unitary $U$ we
already have by the above that $\rho$ satisfies the constraints given by the set $\CM' = \{U^\dagger M_1 U,\ldots,U^\dagger M_m U\}$
and hence $U \rho U^\dagger$ remains a valid $p$-non-local state.

\section{Generalized non-local theories}\label{sec:gnst}

To create a closer analogy between our ``theories'' derived from relaxed uncertainty relations and non-local boxes, we now consider a related class of theories 
called \emph{generalized no-signaling theories} (GNST)~\cite{barrett:nonlocal}, for which we will consider similar relaxations.
As already sketched in the introduction, states in a GNST are defined operationally. Consider a laboratory setup where we have a device which prepares a specific state. We then use a measuring device which has a choice of settings allowing us to measure different properties of the system. The measuring device gives us a reading specifying the outcome of the measurement. A particular state in GNST is described completely by means of the probabilities of obtaining each outcome when
performing a fixed set of {\em fiducial} measurements.
For example, for a set of fiducial measurements $\setM = \{X,Z,Y\}$ with outcomes $\setA = \{\pm 1\}$, the probabilities
$p(A|\CM)$ for all $A \in \setA$ and $\CM \in \setM$ form a description of the state. Hence, we will simply use $p$ to refer to a state given by said conditional probabilities.
The idea behind considering fiducial measurements stems from the idea that there exists a set of measurement choices that suffice to fully describe the system. In classical mechanics, for instance, we can always in principle make a single measurement which outputs all the information necessary to describe a state. For a qubit, on the other hand, we would need results from at least three different incompatible measurement settings, e.g., spin in three orthogonal directions. We refer to~\cite{barrett:nonlocal} for a definition of GNST
and its allowed operations. For us it will only be important to note that similar to the setting of non-local boxes, we can make only one measurement on each system, and there is no real notion of post-measurement states defined.

In the following, we will be interested in the special case of
multi-partite systems where on each system we can perform one of three fiducial measurements with outcomes $\pm 1$. 
Using our notation from section~\ref{probability} we write
the set of realizable experiments for GNST as
$$
\setE_G = \{\forall k \in \{1,2,3\}^n : \{W_{1,k_1},\ldots,W_{n,k_n}\} \},
$$
with $W_{i,k_i}$ denoting a choice of the $k_i$th measurement on the $i$th system.
Later we will connect these measurement choices with Pauli measurements via the relation $W_{i,1} = X_i, W_{i,2} = Z_i, W_{i,3} = X_i Z_i$. 
A key point of this definition will be that the partitioning of measurements into $n$ systems will be fixed.
We also demand that probability distributions should satisfy an independence principle. As we pointed out, when restricted to partitions over disjoint parties, this just reduces to the no-signaling principle. That is, the choice of measurement on one subset of particles can not be used to 
send a signal to a disjoint subset.

In analogy to the quantum setting~\cite{barrett:nonlocal}, we let one gbit refer to a single system on which we can perform our set of fiducial measurements given above.
Our definition of a gbit thereby slightly differs from the definition given in~\cite{barrett:nonlocal}, which only allows two fiducial measurements $X$ and $Z$ on a single
gbit. Yet, in order to compare the hierarchy of GNST-like theories we will construct below to the $p$-box states from above we adopt this slightly more general
definition in analogy to a single qubit in the quantum case.
Note that for the set of measurements $\CM \in \setE_G$ specified above, an $n$-gbit state, specified by $p : \setA^{\times n} \times \CM \rightarrow [0,1]$, 
is in GNST if $p$ satisfies constraints (1), (2), and (3') in section~\ref{probability}.
\begin{examp}
Consider the following state of one particle in GNST (or one gbit):
$$
\begin{array}{ccccc} p( A=+1 | M = X) &=& s_x &=& 1- p( A=-1 | M = X) \\ p(A=+1 | M = Z) &=& s_y &=& 1- p( A=-1 | M = Z) \\ p(A=+1 | M = XZ) &=& s_z &=& 1- p( A=-1 | M = XZ)\end{array}
$$
This state is normalized, and positivity requires $s_x,s_y,s_z \in [0,1]$. 
The state would be equivalent to the state of an arbitrary qubit if and only if $s_x^2 + s_y^2 +s_z^2 \leq 1$, that is, if we are constrained to the Bloch sphere. 
\end{examp}
For multi-partite states the difference between constraints on qubits and gbits becomes more complicated. 
We now turn to describing a hierarchy of constraints on GNST theories which will be analogous to uncertainty conditions in $p$-nonlocal theories and quantum mechanics.

\subsection{$p$-GNST}\label{sec:pgnst}

Even though states in GNST are defined without any particular structure to their measurements embedded, we will now impose a physically motivated structure.
In particular, we will simply \emph{imagine} in analogy to the quantum setting that measurements $X$, $Z$ and $Y$ obey the same anti-commutation relations
as the Pauli matrices $\{X,Z\} = \{Z,Y\} = \{X,Y\} = 0$. In our definition below, we will for simplicity write $\{\cdot,\cdot\}$ to indicate that we 
imagine such an anti-commutation constraint to hold exactly when the string of Paulis $\prod_i W_{i,k_i}$ associated with each $\CM$ would anti-commute.

First of all, this will allow us to artificially impose an uncertainty relation just like Eq.~(\ref{eq:cliffordUR}).
\begin{definition}
A state is in $p$-GNST if it is in GNST and for any set of measurements $S= \{\CM \in \setE_G\}$ 
satisfying that for all $\CM,\CM' \in S$, $\{\CM,\CM'\} = 0$ 
we have
\be\label{eq:pgnst}
\sum_{\CM \in S} |m(\CM)|^p \leq 1.
\ee
\end{definition}
Note that for $p \rightarrow \infty$ this condition no longer restricts the states, because we get $\max_{\CM \in \setS} |m(\CM)| \leq 1$, 
which is true for the original GNST, and non-local boxes. 
If we would actually add such commutation and anti-commutation constraints 
we could now again distinguish between adding the consistency constraints of section~\ref{sec:moments} only for
measurements acting on different systems, or for all commuting measurements in analogy to the $p$-box and $p$-nonlocal theories.
In analogy to GNST, where commutation relations were only defined for measurements acting on different systems however, we will
stick to this setting, even when considering $p < \infty$. A $p$-GNST state is thus essentially analogous to a $p$-box state, except we are allowed to make simultaneous measurements of locally disjoint systems.

\section{Superstrong non-locality}\label{sec:superstrong}

Before we show that relaxing the uncertainty equation of Eq.~(\ref{eq:cliffordUR}) leads to superstrong non-local correlations, let's take a look at what effect this uncertainty relation actually has on quantum strategies for the
CHSH inequality.
For this purpose, we will rewrite Tsirelson's bound for the CHSH inequality in its more common form as
$$
|\langle A_0 \otimes B_0 \rangle + \langle A_0 \otimes B_1 \rangle + 
\langle A_1 \otimes B_0 \rangle - \langle A_1 \otimes B_1 \rangle| \leq 2 \sqrt{2},
$$
where we use $A_0,A_1$ and $B_0, B_1$ to denote Alice's and Bob's observables respectively where $A_0^2 = A_1^2 = B_0^2 = B_1^2 = \id$.
We will use the fact that in order to achieve the maximum possible quantum violation we must have
$\{A_0,A_1\} = 0$ and $\{B_0,B_1\} = 0$~\cite{tsirel:original,toner:monogamy,uffink:antiComm}. For $M_1 = A_0 \otimes B_0$, $M_2 = A_0 \otimes B_1$,
$M_3 = A_1 \otimes B_0$ and $M_4 = A_1 \otimes B_1$ this means that we have
$\{M_1,M_2\} = \{M_1,M_3\} = \{M_2,M_4\} = \{M_3, M_4\} = 0$. Using the uncertainty relation of Eq.~(\ref{eq:cliffordUR})
proving Tsirelson's bound is equivalent to solving the following optimization problem
\begin{sdp}{maximize}{$\langle M_1 \rangle + \langle M_2 \rangle + \langle M_3 \rangle - \langle M_4 \rangle$}
&$\langle M_1\rangle^2 + \langle M_2 \rangle^2 \leq 1$\\
&$\langle M_1\rangle^2 + \langle M_3 \rangle^2 \leq 1$\\
&$\langle M_2\rangle^2 + \langle M_4 \rangle^2 \leq 1$\\
&$\langle M_3\rangle^2 + \langle M_4 \rangle^2 \leq 1$
\end{sdp}
By using Lagrange multipliers, it is easy to see that for the optimum solution we have
$\langle M_1 \rangle^2 = \langle M_4\rangle^2$ and $\langle M_2 \rangle^2 = \langle M_3\rangle^2$.
By considering all different possibilities, we obtain that with
$x = \langle M_1 \rangle = - \langle M_4 \rangle$ and $y = \langle M_2 \rangle = \langle M_3 \rangle$
our optimization problem becomes
\begin{sdp}{maximize}{$2 (x+y)$}
&$x^2 + y^2 \leq 1$
\end{sdp}
Again using Lagrange multipliers, we now have that the maximum is attained at $x = y = 1/\sqrt{2}$
giving us Tsirelson's bound. 

Tsirelson's bound can hence be understood as a consequence of the uncertainty relation of~\cite{ww:cliffordUR}.
Thus, we intuitively expect that relaxing this relation affects the strength of non-local correlations.
In a similar way, one can view monogamy of non-local correlations as a consequence of Eq.~(\ref{eq:cliffordUR})~\cite{gs:mfromur}.

\subsection{CHSH inequality}

\subsubsection{In $p$-theories}

To see what is possible in $p$-theories, we first construct the equivalent of a maximally entangled state.
Let
$$
\rho_p = \frac{1}{2} \left[\id + \left(\frac{1}{2}\right)^\frac{1}{p}(X + Y)\right].
$$
Note that for $p\rightarrow \infty$ this gives us
$$
\rho_\infty = \frac{1}{2} \left[\id + X + Y\right].
$$
We now proceed analogously to the quantum case to construct
$$
\eta_1 = \mbox{CNOT} (\rho_p \otimes \outp{0}{0}) \mbox{CNOT}^\dagger,
$$
which by claim~\ref{claim:operations} is a valid $p$-bin and $p$-nonlocal state. It can also be verified
that $\eta_1$ forms a valid $p$-box state.

\begin{claim}
Let $A_1 = X$, $A_2 = Y$, $B_1 = X$ and $B_2 = Y$ be Alice and Bob's observables respectively.
Then
$$
\langle CHSH_p \rangle = \Tr(\eta_1(A_1 \otimes B_1 + A_1 \otimes B_2 + A_2 \otimes B_1 - A_2 \otimes B_2)) = 4 \frac{1}{2^{1/p}},
$$
for all $p$-theories.
\end{claim}
\begin{proof}
This follows immediately by noting that
$$ 
\eta_1 = \frac{1}{4}\left(\id + \frac{1}{2^{1/p}}\left(X \otimes X + X \otimes Y + Y \otimes X - Y \otimes Y\right) + Z \otimes Z\right).
$$
\end{proof}

We can also phrase this statement in terms of probabilities as stated in the introduction, by noting that the maximum
probability that Alice and Bob win the CHSH game is given by
$$
\frac{1}{2} + \frac{\langle CHSH_p\rangle}{8} = \frac{1}{2} + \frac{1}{2 \cdot 2^{1/p}}.
$$
It is important to note that this violation can be obtained even when imposing the additional consistency constraints from section~\ref{sec:moments}.

\subsubsection{In $p$-GNST}
We already saw in the introduction that GNST admits states analogous to a non-local box, allowing for a maximal violation of the CHSH inequality. 
We now show that similar states exist for $p$-GNST theories analogous to $p$-box states. We first phrase the CHSH inequality in terms
of probabilities.
In particular, consider the GNST state specified by $p((A_1, A_2) | \{M_1, M_2\}) = \frac14(1+ (-1)^{\delta_{M_1,Z_1} \delta_{M_2,Z_2}} A_1 A_2 \lambda)$ for some $\lambda$ to be chosen below. 
If each party measures $X$ or $Z$ on their state and outputs the result $\pm1$, the probability that Alice and Bob win the CHSH game is given by
\benn
\frac14(& p(1,1|X_1,X_2)+ p(-1,-1|X_1,X_2)+p(1,1|X_1,Z_2)+p(-1,-1|X_1,Z_2)\\
+&p(1,1|Z_1,X_2)+p(-1,-1|Z_1,X_2)+p(1,-1|Z_1,Z_2)+p(-1,1|Z_1,Z_2)) = \frac{1+\lambda}{2}
\eenn
In terms of the moments, $m(X_1,X_2)=m(X_1,Z_2)=m(Z_1,X_2)=-m(Z_1,Z_2)=\lambda$, and this becomes
\benn
\frac14(2+\half(m(X_1,X_2)+m(X_1,Z_2)+m(Z_1,X_2)-m(Z_1,Z_2))) = \frac{1+\lambda}{2}
\eenn
Now we can consider the maximum value of $\lambda$ that is a valid state in $p$-GNST. The requirements listed in example~\ref{nlmoments} only restrict $|\lambda | \leq 1$. Eq.~(\ref{eq:pgnst}) requires $|m(X_1,X_2)|^p+ |m(X_1,Z_2)|^p=|m(Z_1,X_2)|^p+ |m(Z_1,Z_2)|^p=2 | \lambda |^p \leq 1 \rightarrow \lambda = (\half)^{\frac1p}$. Therefore in a $p$-GNST it is possible to win the CHSH game with probability $1/2 + 1/(2 \cdot 2^{1/p})$.

\subsection{XOR games}

We now investigate the case of general 2-player XOR-games for $p \rightarrow \infty$.
In such a game we have an arbitrary (but finite) set of questions $S$ and $T$ from which Alice's and Bob's questions
$s \in S$ and $t \in T$ are chosen according to a fixed probability distribution $\pi: S \times T \rightarrow [0,1]$.
Yet, the set of possible answers remain $A = B = \01$ for Alice and Bob respectively. The game furthermore
specifies a predicate $V: A \times B \times S \times T \rightarrow \01$ that determines the winning answers
for Alice and Bob. In an XOR game, this predicate depends only on the XOR $c = a \oplus b$ of Alice's answer $a$ 
and Bob's answer $b$. We thus write $V(c|s,t) = 1$ if and only if answers $a \oplus b$ satisfying $a \oplus b = c$
are winning answers for questions $s$ and $t$. We will also restrict ourselves to unique games, which have the
property that for any $s,t,b$, there exists exactly one winning answer $a$ for Alice (and similarly for Bob).

First of all, note that in the quantum case we may write the probability that Alice and Bob return answers
$a$ and $b$ with $a \oplus b = c$ as
$$
p(c|s,t) = \frac{1}{2}(1 + (-1)^c \bra{\Psi}A_s \otimes B_t\ket{\Psi}),
$$
where we again use $A_s$ and $B_t$ to denote Alice's and Bob's observable corresponding to questions
$s$ and $t$ respectively and $\ket{\Psi}$ denotes the maximally entangled state. Note that
we again have $(A_s)^2 = (B_t)^2 = \id$ from the fact that both measurements
have only two outcomes. The probability that Alice and Bob win the game can then be written as
$$
\sum_{s,t} \pi(s,t) \sum_c V(c|s,t) p(c|s,t).
$$
Let $v_{st} = \bra{\Psi}A_s \otimes B_t\ket{\Psi}$.
First of all note that for $p \rightarrow \infty$
\begin{equation}\label{eq:bigState}
\frac{1}{d}\left(\id + \sum_{st} v_{st} \Gamma_s \otimes \Gamma_t\right)
\end{equation}
with $d = 2^{\max{|S|,|T|}}$ and $\Gamma_s, \Gamma_t$ anti-commuting observables 
as defined in section~\ref{sec:prelim} is a valid state for any $|v_{st}| \leq 1$.
Hence, we can immediately see that

\begin{corollary}
In any $\infty$-theory, there exists a strategy for Alice and Bob to win a unique XOR game with certainty.
\end{corollary}
\begin{proof}
Consider the state given in Eq.~(\ref{eq:bigState}) with $v_{st} = \pm 1$ such that
$p(c|s,t) = 1$ whenever $V(c|s,t) = 1$. Let Alice and Bob's measurements be given by
$\Gamma_s$ and $\Gamma_t$ for questions $s$ and $t$ respectively, which are valid measurements
for all $p$-theories with $\Gamma_s, \Gamma_t$ constructed as in section~\ref{sec:prelim}.
\end{proof}

We leave it as an open question to examine the case of $p < \infty$ for XOR games, since our aim was 
merely to show that superstrong correlations can exist, if we allow for relaxed uncertainty relations. 
We can see that letting $v_{st} = \pm 1/(\max{|S|,|T|})^{1/p}$ makes Eq.~(\ref{eq:bigState}) a valid state for any
choice of $p$, but this may not generally be the optimal choice. 
The case of GNST is similar, and it has been shown that any non-local correlations can (approximately) be simulated by such
boxes~\cite{wolf:universal}. Optimal bounds for $p$-GNST with $p < \infty$ can be obtained using
techniques analogous to~\cite{qmp}.

\section{Superstrong random access encodings}\label{sec:rac}\label{sec:random}

The existence of superstrong non-local correlations is by no means the only difference we can observe when moving from quantum theory to $p$-GNST or $p$-nonlocal theories.
In particular, we now show that we can obtain so-called random access encodings which, depending on the theory, can be exponentially better than those realized by quantum mechanics.
We then investigate how uncertainty relations and the restrictions imposed by simultaneous measurements affect this encoding.
The existence of such random access encodings will play a crucial role when considering the power of $p$-GNST theories for communication complexity in section~\ref{sec:commComplex}.
In section~\ref{sec:learnability} we also use this random access code to prove a lower bound on the sample complexity of learning states in GNST.
 
\subsection{In $p$-GNST}
Intuitively, a random access code~\cite{ambainisnayak,nayak:original} allows us to encode $N$ bits into a physical system of size $n$ such that we can decode any one bit of the original string with probability at least $q$. More formally,
\begin{definition}
A $[N, n, q]$-random access code (RAC) is an encoding of a string $x \in \{ 0,1 \}^N$ into an $n$-gbit state $p_x$, such that there exist measurements 
$\CM \in \setE_G$ with outcomes $A \in \setA^{\times n}$, and a decoding algorithm $D:  \setA^{\times n} \rightarrow \{ 0,1\}$ satisfying 
$$
Pr(D(A) = x_k) =\sum_{A\in \setA^{\times n}} \delta_{D(A),x_k} p_x(A|\CM)  \geq q,
$$
where $p_x(A|\CM)$ is the probability of obtaining outcome $A$ when performing the measurement $\CM$.
\end{definition}
It has been shown~\cite{ambainisnayak,nayak:original} that in the quantum case, we must have $n \geq (1-h(q))N$, where $h$ denotes the binary entropy function.
There also exist classical encodings for which $n = (1 - h(q))N + O(\log N)$~\cite{ambainisnayak}.
Hence, quantum states offer at most a modest advantage over classical mechanics and, for $q=1$, no advantage at all. 
We now proceed to the surprising result that general no-signaling states lead to extremely powerful random access codes.

\begin{claim}\label{perfectRAC}\label{pgnstRAC}
In GNST, there exists a $[3^n ,n,1]$-random access code. 
\end{claim}
\begin{proof}

An $n$ gbit state in GNST is completely characterized by the probabilities of outcomes
for a fixed set of measurements.
Recall that a single gbit is a two-level system on which we allow three possible measurements with two possible outcomes each.
Also recall that each $\CM \in \setE_G$ can be represented as $\setE_G = \{\forall k \in \{1,2,3\}^n : \{W_{1,k_1},\ldots,W_{n,k_n}\} \},$ with $W_{i,1} = X_i, W_{i,2} = Z_i, W_{i,3} = X_i Z_i$.
Note that each measurement $\CM$ is associated with one of $N=3^n$ vectors $k = (k_1,\ldots,k_n)$. Let $f: \CM \rightarrow \{1,\ldots,N\}$ be a one-to-one function.
For each of the $N = 3^n$ bits we wish to encode, we must specify one measurement $\CM$ that we can use to extract the $j$th-bit. Let that measurement be denoted by $f^{-1}(j).$

We are now ready to define our encoding of the string $x \in \{0,1,2\}^N$ into an $n$-gbit GNST state $p_x$ via the probabilities
$$
p_x( A | \CM ) \assign \frac1{2^n} (1+ A^* (-1)^{x_{f(\CM)}}),
$$
where we use the previously defined notation $A^* = \prod_{i=1}^{|\CM|} A_i$.
It is straightforward to verify that the state is normalized, positive, and satisfies the no-signaling condition.

We now show that any bit of the original string can be decoded perfectly. If we choose to retrieve bit $j$, we measure $\CM=f^{-1}(j).$ 
That means that we get result $A$ with probability  $\frac1{2^n} (1+ A^* (-1)^{x_j})=\frac1{2^n} 2 \delta_{A^*,(-1)^{x_j} }$. 
And we get the result $A^*=(-1)^{x_j}$ with probability:
$$\sum_{A^*=(-1)^{x_j}} p_x(A|\CM)=\sum_{A^*=(-1)^{x_j}} \frac1{2^n} 2 \delta_{A^*,(-1)^{x_j} }=1.$$
where the last equality follows from the fact that we sum over exactly half the $2^n$ possible outcomes $A_1,\ldots,A_n$.
Hence the decoder $D(A) = \half( 1-A^* )$ will return $x_j$ with perfect probability.
\end{proof}

What happens if we impose the uncertainty relation in $p$-GNST? For convenience sake, note that we could rewrite the encoding above
in terms of moments, where we let an encoding of a string $x$ be determined by the moment representation of $p_x$ as
$$
m_x(\CM=f^{-1}(k)) \assign (-1)^{x_k}
$$
with all other moments set to $0$. 

To construct an encoding for $p$-GNST, we consider
$$
m_x(\CM=f^{-1}(k)) \assign (-1)^{x_k} \lambda.
$$ 
What's the largest $\lambda$ that satisfies the uncertainty relation? As we noted earlier the maximum number of anti-commuting 
Pauli operators is $2n+1$, so the most restrictive condition we could get from the uncertainty relation is $(2n+1) | \lambda |^p \leq 1$. 
We thus obtain

\begin{claim}\label{pRAC}
In $p$-GNST, there exists a $[3^n ,n, \half + \half \left(\frac{1}{2n+1}\right)^{1/p}]$-random access code. 
\end{claim}
\begin{proof}
Let $\lambda = (2n+1)^{1/p}$, and note that this satisfies the uncertainty relation. Our encoding is now
$$
p_x( A | \CM) = \frac{1}{2^n} (1+ (-1)^{x_{f(\CM)}} ~\lambda~ A^*).
$$
And our probability of getting the correct sign from our measurement goes down to
$$
\Pr ( D(A) = x_k) = 
\frac{1+|\lambda|}{2} = \half + \half \left(\frac{1}{2n+1}\right)^{1/p}
$$
\end{proof}

If $p < \infty$ we get an encoding that gets asymptotically worse for large $n$.
This should be compared to the bound on the number of qubits for a \emph{quantum} random access 
encoding of $N=3^n$ bits into $k$ qubits with recovery probability $q = 1/2 + 1/2 (1/(2n+1))^{1/p}$. 
From the bound of~\cite{ambainisnayak,nayak:original}, we have that the encoding uses exponentially fewer physical bits than what can be obtained
in the quantum setting and hence even $p$-GNST has a powerful coding advantage over quantum mechanics.
Note that we are always free to split the $N$ bits into smaller pieces first, and encode
each piece independently to keep the recovery probability $q$ constant. This is analogous to the quantum setting where we can encode
each 3 bits into one qubit to obtain a random access code with $n = N/3$.
Alternatively, we can form a simple repetition code, where we have $k$ copies of the random access codes constructed above. We then have
\begin{claim}\label{pRAC-rep}
In $p$-GNST, there exists a $[3^n, (2n+1)^{3/p} n, 1-\eps]$-random access code with $\eps=2 \exp(-(2n+1)^{1/p}/2)$.
\end{claim}
\begin{proof}
We take $k$ copies of the RAC defined in Claim~\ref{pRAC}, and decode by taking the majority of the individual encodings.
Let $Y_j = 1$ if the decoding was successful for the $j$-th copy, and $Y_j=0$ otherwise.
From the Hoeffding inequality we immediately obtain that for $Y=\sum_{j=1}^k Y_j$ and $q$ as defined above
$$
\Pr\left[|Y - qk| \geq t~k \right] \leq 2e^{-2 t^2 k},
$$
If we set $t=q-1/2=1/2 (1/(2n+1))^{1/p}$, that gives us $\Pr\left[Y \leq k/2\right] \leq 2 e^{- \half \left(\frac1{2n+1}\right)^{2/p} k}$. Now if we set $k=(2n+1)^{3/p}$, we have used a total of $(2n+1)^{3/p} n$ gbits and will succeed with probability $1-2 e^{-(2n+1)^{1/p}/2}$ as promised.
\end{proof}

Whereas $(2n+1)^{3/p}n$ is still quite large, note that it is nevertheless only polynomial in $n$. The length of the RAC is hence
still poly-logarithmic in our original input size, where we achieve (near) perfect recovery for large $n$.
Finally, we will need to use one more related result.
\begin{claim}\label{pRAC-learn}
In $p$-GNST, for $\gamma \in (0,1/2)$ and $\hat{n} \geq 2^{2/p} \ln(4/(1/2-\gamma)^2)$, there exists a $[3^{{n}({\hat{n}},p,\gamma)}, {\hat{n}}, \half+\gamma]$-random access code with ${{n}}({\hat{n}},p,\gamma) = \lfloor \left(  \frac{{\hat{n}}~ 2^{-2/p}}{\ln(4/(1/2-\gamma)^2)}  \right)^{\frac1{2/p+1}} \rfloor$.
\end{claim}
\begin{proof}
Again we take $k$ copies of the RAC defined in Claim~\ref{pRAC}, and decode by taking the majority of the individual encodings. The probability to decode correctly in that case was $1-2 e^{-\half (\frac1{2{n}+1})^{2/p} k}$. Now we want to adjust $k$ and ${n}$ to get a code with a fixed success rate and that uses no more than $\hat{n}$ gbits. We need that
(i) $k {n} \leq {\hat{n}}$, that is, our encoding uses at most ${\hat{n}}$ physical bits and
(ii) $1-2 e^{-\half (\frac1{2{n}+1})^{2/p} k} \geq 1/2+\gamma$, which forces our probability of success to be at least $1/2+\gamma$.
We can satisfy (ii) if we set $k=\ln(4/(1/2-\gamma)^2) (2 {n}+1)^{2/p}$, then (i) tells us that
$k {n} = \ln(4/(1/2-\gamma)^2) (2 {n}+1)^{2/p} n$, from which we have
$\ln(4/(1/2-\gamma)^2) 2^{2/p} {n}^{2/p+1} \leq k n \leq \hat{n}$ and thus
$$
{n} \leq \left(  \frac{{\hat{n}}~ 2^{-2/p}}{\ln(4/(1/2-\gamma)^2)}  \right)^{\frac1{2/p+1}}.
$$
Since the smallest system we can encode into is ${n}=1$, this tells us that ${\hat{n}}$ must be at least $ 2^{2/p} \ln(4/(1/2-\gamma)^2)$.
\end{proof}
Note that although this may not be the best encoding, it suffices to give us the asymptotic behavior for $\hat{n}$.

\subsection{In $p$-nonlocal theories}

It is instructive to consider such superstrong encodings in the language of $p$-nonlocal theories to see how such
superstrong encodings would look like in terms of Pauli matrices. This will also allow us to compare the consequences of 
restrictions due to the consistencies of moments from section~\ref{sec:moments} to random access encodings.
For the least restrictive $p$-theory, the $p$-bin theory, we can construct the following very simple encoding.

\begin{claim}\label{pbinRAC}
In $p$-bin theories, there exists a $[2^{2n}-1,n, \half + \half \left(\frac{1}{2n+1}\right)^{1/p}]$-random access code.
\end{claim}
\begin{proof}
Consider the encoding of a string $x \in \01^N$
with $N=2^{2n}-1$ into an $n$ $p$-bit state given by
$$
\rho_x \assign \frac{1}{d}\left(\id + \frac{1}{(2n+1)^{1/p}} \sum_{k=1}^{2^{2n}-1} (-1)^x_k S_{k}\right),
$$
where $S_k = S_{ab}$ is a string of Pauli matrices, where we simply relabeled the indices $ab$.
To decode the $k$th-bit, we measure $S_{k}$. A straightforward calculation shows that the probability 
to obtain outcome $x_k$ is given by
$$
\Pr[x_k] = \frac{1}{2}\Tr\left[\left(\id + S_k\right)\rho_x\right] = \frac{1}{2} + \frac{1}{2 (2n+1)^{1/p}},
$$
as promised. Clearly, the uncertainty relation is satisfied.
\end{proof}

Similarly, we obtain the following encoding for $p$-box theories, which is in one-to-one correspondence with the encodings
in $p$-GNST above.
\begin{claim}\label{pnonlocalRAC}
In $p$-box theories, there exists a $[3^n,n, \half + \half \left(\frac{1}{2n+1}\right)^{1/p}]$-random access code.
\end{claim}
\begin{proof}
Our encoding is analogous to the one above, but we restrict ourselves to including only such strings of
Pauli matrices formed by taking tensor products of $\{X,Y,Z\}$, excluding the identity.
\end{proof}
Clearly, we can again obtain an encoding that is poly-logarithmic in the length of the original input analogous
to Claim~\ref{pRAC-rep} that has perfect recovery for large $n$.

\subsection{The effect of consistency}\label{sec:context}

When viewing such encodings in terms of density matrices, it becomes clear why such encodings do not exist in a quantum
setting: all such encodings are in gross violation of the consistency conditions of section~\ref{sec:moments}.
Even when we restrict ourselves to $p=2$, we can obtain such encodings whereas in the quantum case we cannot.
It is interesting to note that for $p=2$, the violation we can obtain for e.g. the CHSH game is exactly the same
as in the quantum setting. Thus it is perfectly possible to have such superstrong encodings, while simultaneously
being restricted to Tsirelson's bound in the CHSH game for a 2 qubit state. 
This clearly shows how limited our $p$-bin, $p$-nonlocal, but also $p$-GNST theories really are. Since GNST is equivalent
to a theory based on non-local boxes, this also shows that considering such boxes is somewhat limiting, and possibly
ignores
some aspect present in quantum theory that are of importance for information processing. 

\section{Implications for information processing}\label{sec:info}

We now turn to a number of interesting implications of $p$-GNST and $p$-theories to information processing. In particular, 
we will see that both allow us to save significantly on the amount of data we need to transmit to solve certain communication problems.
In fact, we will see that there exists a task for which there exists an \emph{exponential} gap between the amount of communication
required when compared with quantum theory. Other information tasks on the other hand become more difficult. We will see that
when trying to learn states approximately we need to perform exponentially more measurements in the case of GNST.

\subsection{Communication complexity}\label{sec:comm}\label{sec:commComplex}

Imagine two (or more) parties, Alice and Bob, who each have an 
input $x \in \01^n$ and $y \in \01^n$ respectively, unknown to the other party. Their goal is to compute a fixed function $f: \01^{2n} \rightarrow \01^m$ by communicating over a channel. The central question of communication complexity is how many bits they need to transmit in order to compute $f$.
Typically, we thereby only require one party (Bob) to learn the result $f(x,y)$.
To help them reduce the amount of communication, Alice and Bob may possess additional resources such as shared randomness, entanglement, non-local boxes or communicate over a quantum
channel, and may impose different measures of success. For example, they could be interested in computing $f$ only with a certain probability 
instead of computing it exactly. It is well-known that if Alice and Bob can share non-local boxes, they can compute any Boolean function $f:\01^{2n} \rightarrow \01$ perfectly by communicating only a single bit~\cite{wim:nonlocal}, which is 
even true when the non-local boxes have slight imperfections~\cite{falk:nonlocal}. 
Here, we consider the case where Alice and Bob have \emph{no} a-priori resources, however, we they are able to exchange $p$-GNST or $p$-nonlocal states over a suitable channel. 

\subsubsection{One-way communication}

We first of all make a very modest statement and show that in \emph{any} one-way communication protocol, where Alice sends a single message to Bob, we are able to save a constant number of bits, when computing a Boolean function $f$. 
These savings are an immediate consequence of the existence of superstrong random access codes that 
we discussed in section~\ref{sec:rac}. To communicate with Bob, Alice constructs the string 
$$
m = f(x,0),\ldots,f(x,2^n-1)
$$
and encodes $m \in \01^{2^n}$ into a random access code $\rho_m$. To retrieve the correct answer, Bob simply retrieves bit
$x_y = f(x,y)$ from $\rho_m$. Evidently, this type of saving is particularly interesting in the case where 
Alice and Bob would need to communicate $n$ bits to compute $f$, which is the case classically and
quantumly if $f=IP$ is the inner product~\cite{ronald:survey}. By Claims~\ref{pgnstRAC}, ~\ref{pRAC},~\ref{pnonlocalRAC} and~\ref{pbinRAC} we immediately obtain
that

\begin{claim}
Let $p \rightarrow \infty$. Then in to compute the inner product Alice needs to transmit at most
$k$ bits to Bob, where
$$
k = \left\{ \begin{array}{cl}
(1/\log 3)n & \mbox{ for } p\mbox{-GNST and } p\mbox{-nonlocal theories}\\
n/2 & \mbox{ for } p\mbox{-bin theory}
\end{array}\right .
$$
\end{claim}

\subsubsection{Private information retrieval}

More striking though are the possibilities of $p$-GNST or $p$-theories for the task of private information retrieval: Here, one (or more) database servers each hold a copy of the database string $x \in \01^n$. A database user should be able to retrieve any bit $x_i$ of his choosing, while
the servers should not learn the desired index $i$. A protocol that satisfies these parameters is the trivial one, where the server simply 
sends the entire string $x$ to the user. The question is thus, whether it is possible to perform this task by communicating less than $n$ bits.
If only a single server is used, it is known that the trivial protocol is optimal and we need to communicate $\Theta(n)$ bits, even if we are allowed
quantum communication~\cite{kerenidis&wolf:pir}. It is clear that the superstrong encodings from above, allow us to beat this bound trivially, 
by asking the server to encode $x$ into a superstrong random access code. Hence we have as an immediate consequence of Claims~\ref{pgnstRAC},~\ref{pRAC-rep}, ~\ref{pbinRAC}, and ~\ref{pnonlocalRAC} we have
\begin{claim}
In any $p$-GNST, $p$-bin, and $p$-box theory, there exists a single server private information retrieval scheme requiring 
$O(\polylog(n))$ bits of communication for large $n$.
\end{claim}

\subsection{Learnability}\label{sec:learnability}

We consider a scenario in which there is an unknown state for which we are trying to learn an approximate description. In particular, imagine some arbitrary probability distribution over possible two-outcome measurements. We are given the expectation value for each measurement in a finite set picked according to this distribution. We then construct an approximate description of the state which agrees with all the expectation values we have observed so far. 
This description is considered to be good if it predicts the correct results for most future measurements drawn from the same distribution. The central question is how many measurement results we need to be able construct a good description.

The existence of strong random access codes has implications for state learning. Aaronson~\cite{aaronson:learn} used an upper bound on the number of bits that can be encoded into an $n$ qubit RAC to upper bound the number of measurements needed to learn an approximate description of an $n$ qubit state. He took solace in the fact that, despite the exponential number of parameters describing a quantum state, a linear (in the number of qubits) number of measurements suffice to learn an approximate description of the state. If an exponential number of measurements were really required, we could never hope to do enough measurements to verify the identity of quantum states of a few hundred particles.

We show the converse for states in $p$-GNST. We use our constructions of random access codes to lower bound the number of measurements needed to learn an approximate description of the state. We find that an exponential number of measurements is required to find such a description and therefore one could never hope to do enough measurements to learn a description of a state with a modest number of particles, even approximately. This holds even for theories where $p=2$ and the violation of the CHSH inequality is the same as for quantum mechanics. This demonstrates an unusually powerful theory which starkly contrasts with quantum mechanics and the $p$-nonlocal theory. 

We begin with a section defining the relevant tools: a definition of the learning scenario, and a measure of state complexity known as the ``fat shattering dimension.'' We then restate a known lower bound on the number of samples needed for learning in terms of the fat shattering dimension. In the next section, we derive lower bounds on learnability for $p$-GNST theories. First, we use our random access codes to lower bound the fat shattering dimension for $p$-GNST states. Then we can use this result to lower bound the number of samples needed to learn $p$-GNST states. 

\subsection{Tools}
We begin by introducing some terminology from statistical learning theory. Let the set $\setS$ denote the sample space, which will correspond to the space of possible measurements in our case. A probabilistic concept over $\setS$ is just a function $F: \setS \rightarrow [0,1]$, and is equivalent to a state which maps measurement choices to expectation values. A set of such concepts is referred to as the concept class $\mathcal{C}$ over $\setS$ and corresponds to the set of all states. We consider the learning situation in which you are given the value of the target concept (state) over some samples drawn independently according to an arbitrary distribution. The goal is to output a hypothesis concept that will give values close to the target concept for most samples drawn from the same distribution. A sample size that is large enough to allow this to be accomplished with high probability is said to be sufficient. 
To restate the connection, in GNSTs we will say that a state corresponds to a concept, and a measurement on the state to a sample. We will make these notions precise before we demonstrate the connection between RACs and fat-shattering dimension in \ref{measurementbound}. 

We adopt our definition of probabilistic concept learning from Anthony and Bartlett\cite{anthonybartlett}.

\begin{definition}
[Anthony and Bartlett \cite{anthonybartlett}]\label{occam}Let $\setS$\ be a
sample space, let $\mathcal{C}$\ be a probabilistic concept class over $\setS$, and
let $\mathcal{D}$ be a probability measure over $\setS$. 
\ Fix an
element $\rho\in\mathcal{C}$, as well as error parameters $\varepsilon
,\eta,\gamma>0$\ with\ $\gamma>\eta$.\ 
Let $k_0(\eta,\gamma,\epsilon,\delta)$ be some function of the error parameters.
\ Suppose we draw a training set of $k$
samples\ $\mathcal{T}=\left(  s_1,\ldots,s_k\right)  $\ independently according to
$\mathcal{D}$, and then choose any hypothesis $\sigma_{\mc{T}} \in\mathcal{C}$\ such that
$\left\vert \sigma_{\mc{T}} \left(  s_i\right)  -\rho\left(  s_i \right)  \right\vert \leq\eta$\ for
all $s_i\in \setS$. \ Then if for $k \geq k_0(\eta,\gamma,\epsilon,\delta)$%
\[
\Pr_{s\in\setS}\left[  \left\vert \sigma_{\mc{T}} \left(  s\right)  -\rho\left(  s \right)
\right\vert >\gamma\right]  \leq\varepsilon
\]
with probability at least $1-\delta$\ over $\mc{T}$, we say that $k_0$ is a sufficient sample size to learn $\mathcal{C}$.

\end{definition}
This says that if the size of the training set, $k$, is bigger than $k_0$, then with probability $1-\delta$, the training set $\mc{T}$, that we pick according to $\mc{D}$ will be a good training set. 
That is, a hypothesis concept $\sigma$ which matches the target state on the training set will only be different from the target state on some other sample with small probability, $\epsilon$.

%Define fat-shattering
To define a lower bound on $k_0$, we will need a measure of complexity called the \emph{fat-shattering dimension}. 
\begin{definition}[Aaronson \cite{aaronson:learn}]\label{fat}
Let $\setS$\ be a sample space, let $\mathcal{C}$ be a probabilistic-concept class
over $\setS$, and let $\gamma>0$ be a real number. \ We say a set
$\left\{  s_1,\ldots,s_k\right\}  \subseteq\setS$\ is $\gamma
$-fat-\textit{shattered} by $\mathcal{C}$ if there exist real numbers
$\alpha_{1},\ldots,\alpha_{k}$ such that for all $B\subseteq\left\{
1,\ldots,k\right\}  $, there exists a probabilistic concept $\rho \in\mathcal{C}$\ such that
for all $i\in\left\{  1,\ldots,k\right\}  $,

\begin{enumerate}
\item[(i)] if $i\notin B$ then $\rho\left(  s_{i}\right)  \leq\alpha_{i}-\gamma$, and

\item[(ii)] if $i\in B$ then $\rho\left(  s_{i}\right)  \geq\alpha_{i}+\gamma$.
\end{enumerate}

Then the $\gamma$-fat-\textit{shattering dimension} of $\mathcal{C}$, or
$\fat_{\mathcal{C}}\left(  \gamma\right)  $, is the
maximum $k$ such that some $\left\{  s_{1},\ldots,s_{k}\right\}
\subseteq\setS$ is $\gamma$-fat-shattered by $\mathcal{C}$. \ (If there
is no finite such maximum, then $\fat_{\mathcal{C}}\left(  \gamma\right)  =\infty$.)
\end{definition}

The fat-shattering dimension lower bounds the number of samples needed to learn a probabilistic concept.
\begin{lemma}[Anthony and Bartlett \cite{anthonybartlett}]\label{measurementbound}
Suppose $\mathcal{C}$ is a probabilistic concept class over $\setS$ and set $0<\gamma<\eta<1, \epsilon,\delta \in (0,1)$. Then if $\fat_{\mathcal{C}}(\gamma) \geq d \geq 1$ and $\gamma^2 \geq 4 d 2^{-\sqrt{d/6}}$, any sample size $m_0$ sufficient to learn $\mathcal{C}$ satisfies
$$
m_0(\eta,\gamma,\epsilon,\delta) \geq max\left( \frac1{32 \epsilon} \left( \frac{d}{2 \ln^2 (4d/\gamma^2)}-1\right) ,\frac1{\epsilon} \ln \frac1{\delta} \right)
$$
\end{lemma}
This concludes the results we will need from statistical learning theory.

\subsection{Lower bounds on sample complexity}

Our next step is to show that the existence of random access codes lower bounds the fat-shattering dimension. First we have to carefully define what ``concept'' we will be learning and what constitutes our sample space. 
For the purposes of learning in GNSTs, the sample space is just 
the set of possible measurements, where we allow general measurements by first making some fiducial measurement on the state, and then post-processing the result using some decoding function. So we can define $\mc{S}_{\textrm GNST} \assign \{ (\CM,D) | \CM \in \setE_G, D: \setA^{\times n} \rightarrow \{ 0,1\} \}$.
 For some sample $(\CM,D) \in \mc{S}_{GNST}$, a concept is specified by the state $\rho_x$ in a GNST via the the probability $\rho_x (\CM,D) \assign \sum_{A\in\setA^{\times n}} D(A) p_x(A|\CM)$, where $p_x$ is an $n$-partite state in some GNST. Then the concept class $\mc{C}_{\textrm GNST}$ is the set of concepts specified by all the states in GNST. 
 
Note that a ``sample'' is stronger than a typical notion of measurement. Usually we say that the measurement gives a result with some probability, but given some sample, the concept $\rho$ actually returns the probability of that outcome occurring. This stronger notion of sampling is all we consider here since we are only lower bounding the number of samples needed.

\begin{claim}\label{fatbound}
Let the concept class $\mathcal{C}_{\textrm GNST}$ over $\setS_{\textrm GNST}$ consist of all $\rho_x(\CM,D)= \sum_{A\in\mathcal{A}} D(A) p_x(A|\CM)$, where $p_x$ describes any $n$-partite states in a GNST, over the sample space $ \{ (\CM,D) | \CM \in \setE_G, D: \mathcal{A}^n \rightarrow \{ 0,1\} \}$. For integers $n, N(p,n)$ and $\gamma \in (0,1)$, if there exists an $[N(p,n),n,\half + \gamma]$-RAC then $\fat_{\mathcal{C}_{\textrm GNST}}(\gamma) \geq N$.
\end{claim}
\begin{proof}
By the RAC definition, there exist a set of measurements $\{(\CM,D),\ldots, (\CM^{(N)},D^{(N)})\}$ and states specified by (the concepts) $\rho_x$ for $x\in\{0,1\}^N$ so that
\begin{enumerate}
\item[(i)] if $x_i = 0$ then $\rho_x(\CM^{(i)},D^{(i)}) \leq \half - \gamma$
\item[(ii)] if $x_i = 1$ then $\rho_x(\CM^{(i)},D^{(i)}) \geq \half + \gamma$
\end{enumerate}
Therefore, this set of samples is $\gamma$ fat-shattered by $\mathcal{C}_{\textrm GNST}$. Since $\fat_{\mathcal{C}_{\textrm GNST}}$ is the size of the largest sample set shattered, $\fat_{\mathcal{C}_{\textrm GNST}} \geq N(p,n)$.
\end{proof}

Combining Claims~\ref{pRAC-learn} with \ref{fatbound} and \ref{measurementbound}, we get the following result.
\begin{corollary}
For $\hat{n}$-partite concepts in $\mathcal{C}_{\textrm p-GNST}$ and error parameters $\varepsilon
,\eta,\gamma,\delta>0$\ with\ $\gamma>\eta$, if $\hat{n} \geq 2^{2/p} \ln(4/(1-\gamma)^2)$ and
$$k<max\left( \frac1{32 \epsilon} \left( \frac{3^{{n}(\hat{n},p,\gamma)}}{2 \ln^2 (4\cdot3^{{n}(\hat{n},p,\gamma)}/\gamma^2)}-1\right) ,\frac1{\epsilon} \ln \frac1{\delta} \right)$$ 
for ${{n}}(\hat{n},p,\gamma) = \lfloor \left(  \frac{\hat{n}~ 2^{-2/p}}{\ln(4/(1-\gamma)^2)}  \right)^{\frac1{2/p+1}} \rfloor$, then $k$ is not a sufficient sample size to learn states in $\mathcal{C}_{\textrm p-GNST}$.
\end{corollary}

That is, we need $\bigoh(3^{\hat{n}^{\frac1{2/p+1}}} /{\hat{n}}^{\frac2{2/p+1}} )$ samples to learn an $\hat{n}$-partite state in $p$-GNST to great accuracy. For $p=2$ we have an uncertainty relation analogous to quantum mechanics that rules out super-quantum violations of the CHSH bound. Nevertheless it still takes $\bigoh(3^{\sqrt{\hat{n}}}/{\hat{n}})$ samples to learn these states, as compared to $\bigoh(n)$ in the quantum case. 

\newcommand{\lb}{\raisebox{-.9ex}}
\begin{table}[htdp]\label{table:results}
\begin{center}\begin{tabular}{|c|c|c|c|c|c|}\hline  & p-bin & p-GNST/p-box & p-nonlocal & Quantum & Classical \\
\hline\hline Non-signaling & yes & yes & yes & yes & yes \\
\hline Satisfies p-uncertainty & yes & yes & yes & p=2 & n/a \\\hline Simultaneous  & \lb{no} & \lb{local} & \lb{commuting} & \lb{commuting} & \lb{all} \\[-1ex] measurements&&&&&\\\hline \hline  CHSH violation& \lb{$\half+\frac1{2^{1/p+1}}$} & \lb{$\half+\frac1{2^{1/p+1}}$} & \lb{$\half+\frac1{2^{1/p+1}}$} & \lb{$\half+\frac1{2^{1/2+1}}$} & \lb{$\frac34$} \\[2ex] \hline 
RAC  bits to&\lb{$\bigoh(\polylog(N))$}&\lb{$\bigoh(\polylog(N))$}&\lb{$?$}&\lb{$\Omega(N)$}&\lb{$\Omega(N)$}\\ \raisebox{0.5ex}{encode N bits} &&&&&\\\hline 
PIR from N bits&$\bigoh(\polylog(N))$&$\bigoh(\polylog(N))$&?&$\Omega(N)$&$\Omega(N)$\\\hline``Learning'' states&hard&hard&?&easy&easy\\\hline \end{tabular} \caption{Summary of properties and results for various theories.}
\end{center}
\label{defaulttable}
\end{table}

\section{Consistency of measurements}

\section{Conclusion and open questions}

We have shown that relaxing uncertainty relations can lead to superstrong non-local correlations. This is quite intuitive when considering Tsirelson's bound as
a consequence of such an uncertainty relation in the quantum setting. We then constructed a range of theories inspired by such relaxations, and investigated
their power with respect to a number of information processing problems. In particular, we obtained superstrong random access encodings and savings for communication
complexity. At the same time, however, it turned out to become harder to learn a state in such a theory. We then discussed what makes such superstrong encodings
possible in our $p$-theories, but also in GNST. We identified a number of simple constraints that prevent us from constructing a similar encoding in the
quantum setting. Our work may indicate that using ``box-world'' to understand any other problems within quantum information beyond non-local correlations
may be difficult, as ``box-world'' differs from the quantum setting with respect to such constraints, at least when drawing a one-to-one analogy from a gbit to a qubit as in GNST~\cite{barrett:nonlocal}. It is important to note that these constraints did not
prevent us from observing superstrong non-local correlations, but merely forbid our encodings in section~\ref{sec:random}. 
If one would like to use ``box-world'' to understand other aspects one could either impose such consistency constraints, or look for a different approach to defining
such theories. GNST was defined by first specifying states and then allowing all operations that take valid states to valid states. If one would have specified
the theory in terms of allowed transformations, instead of states, such encodings could also have been ruled out. For example, in the quantum setting one
can transform operators $X \otimes X$, $Z \otimes Z$ and $XZ \otimes XZ$ into a bipartite form via a unitary operation. When looking at a density matrix
expressed in terms of strings of Pauli matrices, its coefficients (which directly determine the moments for measurements of strings of Paulis) must obey
similar constraints to the coefficients belonging to bipartite operators of the form $\id \otimes X, X\otimes \id, X \otimes X$ for example. 

Finally, it is clear that both the uncertainty relation and the consistency constraints are obeyed in the quantum setting, since we demand that for 
any $\rho$ we have $\Tr(\rho) = 1$ and $\rho \geq 0$ to be a valid quantum state. Not surprisingly, both forms of constraints are thus necessary 
(but in higher dimensions not always sufficient) conditions for $\rho \geq 0$. Such characterizations are not easy for $d > 2$~\cite{kimura, bloch2, dietz:bloch,wehner:bloch}, 
and it remains an interesting open problem to find an intuitive interpretation for such conditions in higher dimensions, and their consequence for
information processing tasks.

\section{Acknowledgments}

We are indebted to Wim van Dam for useful discussions. This work was supported by NSF grant number PHY-04056720.


\begin{thebibliography}{10}

\bibitem{aaronson:learn}
S.~Aaaronson.
\newblock {The learnability of quantum states}.
\newblock {\em Royal Society of London Proceedings Series A}, 463:3089--3114,
  December 2007.

\bibitem{ambainisnayak}
A.~Ambainis, A.~Nayak, A.~Ta-Shma, and U.~Vazirani.
\newblock Dense quantum coding and a lower bound for 1-way quantum automata.
\newblock In {\em Proceedings of STOC '99}, pages 376--383, New York, NY, USA,
  1999. ACM.

\bibitem{nayak:original}
A.~Ambainis, A.~Nayak, A.~{Ta-Shma}, and U.~Vazirani.
\newblock Quantum dense coding and a lower bound for 1-way quantum finite
  automata.
\newblock In {\em Proceedings of 31st ACM STOC}, pages 376--383, 1999.
\newblock quant-ph/9804043.

\bibitem{anthonybartlett}
M.~Anthony and P.~L. Bartlett.
\newblock Function learning from interpolation.
\newblock {\em Comb. Probab. Comput.}, 9(3):213--225, 2000.

\bibitem{barnum:broadcasting}
H.~{Barnum}, J.~{Barrett}, M.~{Leifer}, and A.~{Wilce}.
\newblock {Generalized No-Broadcasting Theorem}.
\newblock {\em Physical Review Letters}, 99(24):240501--+, December 2007.

\bibitem{barnum:teleport}
H.~Barnum, J.~Barrett, M.~Leifer, and A.~Wilce.
\newblock Teleportation in general probabilistic theories, 2008.

\bibitem{ben:crypto}
H.~Barnum, O.~Dahlsten, M.~Leifer, and B.~Toner.
\newblock Nonclassicality without entanglement enables bit commitment.
\newblock arXiv:0803.1264, 2008.

\bibitem{barrett:nonlocal}
J.~{Barrett}.
\newblock {Information processing in generalized probabilistic theories}.
\newblock {\em Physical Review A}, 75(3):032304--+, March 2007.

\bibitem{bell:epr}
J.~S. Bell.
\newblock On the {E}instein-{P}odolsky-{R}osen paradox.
\newblock {\em Physics}, 1:195--200, 1965.

\bibitem{bloch2}
R.~A. Bertlmann and P.~Krammer.
\newblock Bloch vectors for qudits.
\newblock {\em Journal of Physics A: Math. Theor.}, 41:235303, 2008.

\bibitem{falk:nonlocal}
G.~Brassard, H.~Buhrman, N.~Linden, A.~Methot, A.~Tapp, and F.~Unger.
\newblock A limit on nonlocality in any world in which communication complexity
  is not trivial.
\newblock {\em Physical Review Letters}, 96:250401, 2006.

\bibitem{kitaev:magic}
S.~Bravyi and A.~Kitaev.
\newblock Universal quantum computation with ideal clifford gates and noisy
  ancillas.
\newblock {\em Physical Review A}, 71:022316, 2005.

\bibitem{wehner05b}
H.~Buhrman, M.~Christandl, F.~Unger, S.~Wehner, and A.~Winter.
\newblock Implications of superstrong nonlocality for cryptography.
\newblock {\em Proceedings of the Royal Society A}, 462(2071):1919--1932, 2006.
\newblock quant-ph/0504133.

\bibitem{tsirel:original}
B.~Tsirelson (Cirel'son).
\newblock Quantum generalizations of {B}ell's inequality.
\newblock {\em Letters in Mathematical Physics}, 4:93--100, 1980.

\bibitem{tsirel:separated}
B.~Tsirelson (Cirel'son).
\newblock Quantum analogues of {B}ell inequalities: The case of two spatially
  separated domains.
\newblock {\em Journal of Soviet Mathematics}, 36:557--570, 1987.

\bibitem{chsh:nonlocal}
J.~Clauser, M.~Horne, A.~Shimony, and R.~Holt.
\newblock Proposed experiment to test local hidden-variable theories.
\newblock {\em Physical Review Letters}, 23:880--884, 1969.

\bibitem{serge:bounded}
I.~Damgaard, S.~Fehr, L.~Salvail, and C.~Schaffner.
\newblock Cryptography in the {B}ounded {Q}uantum-{S}torage {M}odel.
\newblock In {\em Proceedings of 46th IEEE FOCS}, pages 449--458, 2005.

\bibitem{serge:new}
I.~Damg{\aa}rd, S.~Fehr, R.~Renner, L.~Salvail, and C.~Schaffner.
\newblock A tight high-order entropic quantum uncertainty relation with
  applications.
\newblock In {\em Advances in Cryptology---CRYPTO~'07}, volume 4622 of {\em
  Lecture Notes in Computer Science}, pages 360--378. Springer-Verlag, 2007.

\bibitem{dariano}
G.~M. D'Ariano.
\newblock Probabilistic theories: what is special about quantum mechanics?,
  2008.

\bibitem{ronald:survey}
R.~de~Wolf.
\newblock Quantum communication and complexity.
\newblock {\em Theoretical computer science}, 287(1):337--353, 2002.

\bibitem{dietz:bloch}
K.~Dietz.
\newblock Generalized bloch spheres for m-qubit states.
\newblock {\em Journal of Physics A: Math. Gen.}, 36(6):1433--1447, 2006.

\bibitem{qmp}
A.~Doherty, Y.-C. Liang, B.~Toner, and S.~Wehner.
\newblock The quantum moment problem and bounds on entangled multi-provers
  games.
\newblock In {\em Proceedings of the 23rd IEEE Conference on Computational
  Complexity}, pages 199--210, 2008.

\bibitem{wolf:universal}
M.~Forster and S.~Wolf.
\newblock The universality of non-local boxes.
\newblock In {\em Proceedings of QCMC}, 2008.

\bibitem{hardy:axioms}
L.~Hardy.
\newblock {Quantum Theory From Five Reasonable Axioms}.
\newblock 2001.

\bibitem{kerenidis&wolf:pir}
I.~Kerenidis and R.~{de} Wolf.
\newblock Exponential lower bound for 2-query locally decodable codes via a
  quantum argument.
\newblock In {\em Proceedings of 35th ACM STOC}, pages 106--115, 2003.
\newblock quant-ph/0208062.

\bibitem{kimura}
G.~Kimura.
\newblock The bloch vector for n-level systems.
\newblock {\em Physical Review A}, 315:339, 2003.

\bibitem{imai}
G.~Kimura, T.~Miyadera, and H.~Imai.
\newblock Optimal state discrimination in generic probability models, 2008.

\bibitem{masanes:nonsignaling}
L.~{Masanes}, A.~{Acin}, and N.~{Gisin}.
\newblock {General properties of nonsignaling theories}.
\newblock {\em Physical Review A}, 73(1):012112--+, January 2006.

\bibitem{nielsen&chuang:qc}
M.~A. Nielsen and I.~L. Chuang.
\newblock {\em Quantum Computation and Quantum Information}.
\newblock Cambridge University Press, 2000.

\bibitem{peres:book}
A.~Peres.
\newblock {\em Quantum Theory: Concepts and Methods}.
\newblock Kluwer Academic Publishers, 1993.

\bibitem{popescu:nonlocal}
S.~Popescu and D.~Rohrlich.
\newblock Quantum nonlocality as an axiom.
\newblock {\em Foundations of Physics}, 24(3):379--385, 1994.

\bibitem{popescu:nonlocal2}
S.~Popescu and D.~Rohrlich.
\newblock Nonlocality as an axiom for quantum theory.
\newblock In {\em The dilemma of Einstein, Podolsky and Rosen, 60 years later:
  International symposium in honour of Nathan Rosen}, 1996.

\bibitem{popescu:nonlocal3}
S.~Popescu and D.~Rohrlich.
\newblock Causality and nonlocality as axioms for quantum mechanics.
\newblock In {\em Proceedings of the Symposium of Causality and Locality in
  Modern Physics and Astronomy: Open Questions and Possible Solutions}, 1997.

\bibitem{uffink:antiComm}
M.~Seevinck and J.~Uffink.
\newblock Local commutativity versus bell-inequality violation for entangled
  states and versus non-violation for separable states.
\newblock {\em Physical Review A}, 76:042105, 2007.

\bibitem{gs:mfromur}
G.~Ver Steeg and S.~Wehner.
\newblock Monogamy of non-local correlations from an uncertainty relation.
\newblock Unpublished note, 2008.

\bibitem{summers:qftIndep}
S.J. Summers.
\newblock On the independence of local algebras in quantum field theory.
\newblock {\em Reviews in Mathematical Physics}, 2(2):201--247, 1990.

\bibitem{toner:mono1}
B.~Toner.
\newblock Monogamy of nonlocal quantum correlations.
\newblock quant-ph/0601172, 2006.

\bibitem{toner:monogamy}
B.~Toner and F.~Verstraete.
\newblock Monogamy of bell correlations and tsirelson's bound, 2006.
\newblock quant-ph/0611001.

\bibitem{wim:nonlocal}
W.~van Dam.
\newblock Impossible consequences of superstrong nonlocality.
\newblock quant-ph/0501159, 2005.

\bibitem{wainwrightjordan}
M.~J. Wainwright and M.~I. Jordan.
\newblock Graphical models, exponential families, and variational inference.
\newblock Technical report, Dept. of Statistics, September 2003.

\bibitem{wehner:bloch}
S.~Wehner.
\newblock Unpublished note.
\newblock 2008.

\bibitem{ww:cliffordUR}
S.~Wehner and A.~Winter.
\newblock Higher entropic uncertainty relations for anti-commuting observables.
\newblock {\em Journal of Mathematical Physics}, 49:062105, 2008.

\bibitem{WW08:compose}
S.~Wehner and J.~Wullschleger.
\newblock Composable security in the bounded-quantum-storage model.
\newblock In {\em ICALP 2008}, pages 604--615, 2008.

\bibitem{wolf:personal}
S.~Wolf.
\newblock Personal communication.
\newblock 2008.

\bibitem{PRbox}
J.~{Barrett}, N.~ {Linden},S.~ {Massar}, S.~ {Pironio}, S.~ 
	{Popescu},and D.~ {Roberts}.
\newblock Nonlocal correlations as an information-theoretic resource
\newblock {\em Physical Review A}, 71:022101, 2005.

\end{thebibliography}
\end{document}